\documentclass[fontsize=11pt, a4paper, oneside, DIV=15, toc=listofnumbered, toc=bibnumbered, numbers = noendperiod]{scrartcl}

\usepackage[utf8]{inputenc}
\usepackage[T1]{fontenc}
\usepackage{lmodern}

\usepackage{newunicodechar}
\newunicodechar{ρ}{\rho}
\newunicodechar{σ}{\sigma}
\newunicodechar{λ}{\lambda}
\newunicodechar{Λ}{\Lambda}
\newunicodechar{μ}{\mu}
\newunicodechar{ν}{\nu}
\newunicodechar{ψ}{\psi}
\newunicodechar{ϕ}{\phi}
\newunicodechar{φ}{\varphi}
\newunicodechar{π}{\pi}
\newunicodechar{α}{\alpha}
\newunicodechar{ϵ}{\epsilon}
\newunicodechar{ε}{\varepsilon}
\newunicodechar{δ}{\delta}
\newunicodechar{ω}{\omega}
\newunicodechar{∗}{\textasteriskcentered}

\usepackage{microtype}

\usepackage[english]{babel}
\usepackage{csquotes}

\usepackage{xparse}
\usepackage{xspace}
\usepackage{makecell}
\usepackage{tablefootnote}
\usepackage{titling}
\usepackage{appendix}

\usepackage{versions}
\includeversion{DRAFT}
\excludeversion{TRASH}
\excludeversion{EXCLUDED}
\includeversion{TQCE}

\usepackage{amssymb}
\usepackage{amsmath}
\usepackage{mathrsfs}
\usepackage{mathtools}

\usepackage{cool}
\usepackage{xcolor}
\usepackage{hyperref}
\hypersetup{
	colorlinks,
	linkcolor={red!35!black},
	citecolor={red!35!black},
	urlcolor={blue!80!black},
	linktoc=all
}
\usepackage[nameinlink]{cleveref}

\usepackage{amsthm}
\usepackage{thmtools}
\newtheorem{theorem}{Theorem}
\newtheorem{lemma}[theorem]{Lemma}

\newtheorem{proposition}[theorem]{Proposition}
\newtheorem{example}[theorem]{Example}
\newtheorem{definition}[theorem]{Definition}
\theoremstyle{definition}
\newtheorem{remark}[theorem]{Remark}

\newtheorem*{definition*}{Definition}

\usepackage{graphicx}       
\setkeys{Gin}{width=\textwidth} %
\usepackage[multidot,spae]{grffile} %% Enables \includegraphics to parse files with multiple dots and spaces
\usepackage[compatibility=false]{caption}
\usepackage{subcaption}
\graphicspath{{./}{img/}}

\makeatletter
\let\x@caption\caption% original (caption package) \caption
\newcommand{\x@@caption}[2][\empty]{%
	\ifx\empty#1\relax\x@caption{#2}%
	\else\x@caption[#1]{\textbf{#1.} #2}%
	\fi}% modified caption
\let\caption\x@@caption% new global \caption for other float types
\renewcommand{\captionabove}[2][\empty]{\captionsetup{position=above}%
	\x@@caption[#1]{#2}%
}% new \captionabove
\renewcommand{\captionbelow}[2][\empty]{\captionsetup{position=below}%
	\x@@caption[#1]{#2}%
}% new \captionbelow
\makeatother

\usepackage[style=alphabetic, maxbibnames=99]{biblatex} %minalphanames = 1, maxalphanames = 3

\DeclareFieldFormat{citehyperref}{%
	\DeclareFieldAlias{bibhyperref}{noformat}% Avoid nested links
	\bibhyperref{#1}}

\DeclareFieldFormat{textcitehyperref}{%
	\DeclareFieldAlias{bibhyperref}{noformat}% Avoid nested links
	\bibhyperref{%
		#1%
		\ifbool{cbx:parens}
		{\bibcloseparen\global\boolfalse{cbx:parens}}
		{}}}

\savebibmacro{cite}
\savebibmacro{textcite}

\renewbibmacro*{cite}{%
	\printtext[citehyperref]{%
		\restorebibmacro{cite}%
		\usebibmacro{cite}}}

\renewbibmacro*{textcite}{%
	\ifboolexpr{
		( not test {\iffieldundef{prenote}} and
		test {\ifnumequal{\value{citecount}}{1}} )
		or
		( not test {\iffieldundef{postnote}} and
		test {\ifnumequal{\value{citecount}}{\value{citetotal}}} )
	}
	{\DeclareFieldAlias{textcitehyperref}{noformat}}
	{}%
	\printtext[textcitehyperref]{%
		\restorebibmacro{textcite}%
		\usebibmacro{textcite}}}

\AtEveryBibitem{
	\clearfield{note}
	\clearfield{issn}
	\clearfield{urlyear}
	\clearfield{urlmonth}
	\clearfield{eprintclass}
	\iffieldundef{journaltitle}{}{\clearfield{version}}
	\iffieldundef{journaltitle}{\iffieldundef{eprint}{}{\clearfield{doi}}}{}
}

\renewbibmacro*{url}{%
	\iffieldundef{doi}{\iffieldundef{eprint}{%
			\printfield{url}%
		}{}}{}
}

\DeclareBibliographyCategory{cited}
\AtEveryCitekey{\addtocategory{cited}{\thefield{entrykey}}}

\usepackage{suffix}
\usepackage{xparse}
\usepackage{braket}

\let\bkbraket\braket

\usepackage{physics}

\renewcommand{\tr}{\Tr}

\usepackage{mathtools}
\usepackage{amssymb}
\usepackage{tensor}
\usepackage[binary-units=true]{siunitx}

\newcommand{\br}[1]{\left(#1\right)}

\NewDocumentCommand{\stared}{m m}{\IfBooleanTF{#2}{#1*}{#1}}
\NewDocumentCommand{\es}{m o}{
	\IfNoValueTF{#2}{
		\mathbb{\uppercase{#1}}^{\lowercase{#1\/}}
	}{
		\mathbb{\uppercase{#1}}^{#2\/}
	}
}

\newcommand{\reals}{\mathbb{R}}

\newcommand{\naturals}{\mathbb{N}}

\newcommand{\HS}{\mathcal{H}}

\NewDocumentCommand{\BO}{o}{\mathcal{B}\br{\IfValueTF{#1}{#1}{\HS}}}

\NewDocumentCommand{\DM}{o}{\mathcal{D}\br{\IfValueTF{#1}{#1}{\HS}}}

\NewDocumentCommand{\pos}{o}{\mathscr{P}\br{\IfValueTF{#1}{#1}{\HS}}}
\NewDocumentCommand{\supp}{m}{\textrm{supp}\br{#1}}

\DeclarePairedDelimiterX{\infdivx}[2]{(}{)}{
	#1\;\delimsize\|\;#2
}

\DeclareDocumentCommand{\qrd}{o o m m}{\ensuremath{\IfValueTF{#1}{#1}{D}\IfValueTF{#2}{_{#2}}{}\/\infdivx*{#3}{#4}}}
\NewDocumentCommand{\ent}{d()g}{\ensuremath{H\br{\IfValueTF{#1}{#1}{#2}}}}
\NewDocumentCommand{\qent}{d()g}{\ensuremath{H\br{\IfValueTF{#1}{#1}{#2}}}}

\newcommand*\ric[1]{\vphantom{#1}\smash{#1_{}\kern-\scriptspace}}

\newcommand{\renyi}{R\'enyi\xspace}

\newlength\oversetwidth
\newlength\underwidth
\newcommand\alignedoverset[2]{
	\settowidth\oversetwidth{$\overset{#1}{#2}$}
	\settowidth\underwidth{$#2$}
	\setlength\oversetwidth{\oversetwidth-\underwidth}
	\hspace{.5\oversetwidth}
	&
	\settowidth\oversetwidth{$\overset{#1}{#2}$}
	\settowidth\underwidth{$#2$}
	\setlength\oversetwidth{\oversetwidth-\underwidth}
	\hspace{-.5\oversetwidth}
	\overset{#1}{#2}
}
\newcommand{\Sn}{\mathfrak{S}_n}

\newcommand{\n}{^{\otimes n}}
\newcommand{\Pa}{P_A(\pi)}
\newcommand{\Pad}{\Pa^\dagger}
\newcommand{\Pb}{P_B(\pi)}
\newcommand{\Pbd}{\Pb^\dagger}
\newcommand{\divarg}[2]{\infdivx*{#1}{#2}}

\DeclareMathOperator*{\liminfsup}{\overline{\underline{\lim}}}%
\DeclareMathOperator*{\essentiallylim}{``\lim ''}%

\bibliography{bibliography.bib}

\title{Composite Classical and Quantum Channel Discrimination}
\author{Bjarne Bergh\thanks{Department of Applied Mathematics and Theoretical Physics, University of Cambridge, United Kingdom}\and Nilanjana Datta\thanksmark{1} \and Robert Salzmann\thanksmark{1}}

\usepackage{bm}
\renewcommand{\E}{\mathcal{E}}
\newcommand{\F}{\mathcal{F}}

\newcommand{\id}{\textrm{id}}
\renewcommand{\S}{\mathcal{S}}
\newcommand{\T}{\mathcal{T}}
\newcommand{\conv}{\mathcal{C}}
\newcommand{\C}{\conv}
\NewDocumentCommand{\cptp}{o}{\mathrm{CPTP}(\IfValueTF{#1}{#1}{A \to B})}
\newcommand{\M}{\mathcal{M}}

\newcommand{\X}{\mathcal{X}}
\newcommand{\Y}{\mathcal{Y}}
\renewcommand{\D}{\mathbf{D}}

\begin{document}
	
	\maketitle
	\begin{abstract}
		We study the problem of binary composite channel discrimination in the asymmetric setting, where the hypotheses are given by fairly arbitrary sets of channels, and samples do not have to be identically distributed. In the case of quantum channels we prove: $(i)$ a characterization of the Stein exponent for parallel channel discrimination strategies and $(ii)$ an upper bound on the Stein exponent for adaptive channel discrimination strategies. We further show that already for classical channels this upper bound can sometimes be achieved and be strictly larger than what is possible with parallel strategies. Hence, there can be an advantage of adaptive channel discrimination strategies with composite hypotheses for classical channels, unlike in the case of simple hypotheses. Moreover, we show that classically this advantage can only exist if the sets of channels corresponding to the hypotheses are non-convex. As a consequence of our more general treatment, which is not limited to the composite i.i.d.~setting, we also obtain a generalization of previous composite state discrimination results. 
	\end{abstract}
	
	\tableofcontents
	\section{Introduction and Outline}
	
	Hypothesis testing, or finding optimal strategies and minimal errors for discrimination tasks is one of the oldest and most studied tasks in information theory. In the quantum setting there have been plenty of results regarding the optimal discrimination of states~\cite{helstrom_quantum_1976, kholevo_asymptotically_1979, hiai_proper_1991, ogawa_strong_2000, nagaoka_converse_2006, audenaert_discriminating_2007, hayashi_error_2007, audenaert_asymptotic_2008, nussbaum_chernoff_2009, audenaert_quantum_2012, bae_quantum_2015, mosonyi_error_2021} and also quantum channels~\cite{chiribella_memory_2008,duan_perfect_2009,hayashi_discrimination_2009,harrow_adaptive_2010, wang_resource_2019, fang_chain_2020, wilde_amortized_2020}. However, most of the quantum literature focuses on the case where the two hypotheses are simple, i.e.~the hypotheses state that what we are given is \emph{exactly} one specific state (or channel). Arguably, much more practically relevant is the case where one allows for composite hypotheses, i.e.~hypotheses stating that the given state (or channel) belongs to a certain set. 
	This first of all includes the noisy regime, where we can assume that what we are given is approximately one of two possibilities, but also much more general settings, i.e.~questions of discriminating big sets with a certain structure (for states this could for example be sets of separable~\cite{brandao_adversarial_2020} or coherent~\cite{berta_composite_2021} states). 
	
	Throughout this paper we will be looking at binary composite hypothesis testing in the asymmetric setting. We are given $n$ instances of an unknown object, and have to make a decision between two hypotheses based on these $n$ instances, and are ultimately interested in the asymptotic limit $n \to \infty$. We will start with introducing the problem for discriminating two sets of states, and give an overview of previous results in the literature, before moving on to the problem of discriminating two sets of channels. To our knowledge, the task of composite binary quantum {\em{channel}} discrimination has not been studied thus far. Throughout our analysis, we will not restrict ourselves to the composite i.i.d.~setting, i.e.~we will also allow the provided objects (states or channels) to vary within the sets corresponding to the hypotheses.
	
	For binary asymmetric composite channel discrimination we show in this fairly general setting: $(i)$ a characterization of the Stein exponent for parallel channel discrimination strategies (\autoref{thm:parr}), and $(ii)$ an upper bound on the Stein exponent for adaptive channel discrimination strategies (\autoref{prop:adaptive_upper_bound}). We further show that already classically 
	this upper bound can sometimes be achieved and be strictly larger than what is possible with parallel strategies (\autoref{example:adaptive_advantage}), and hence there can be an advantage of adaptive channel discrimination strategies with composite hypotheses. We go on to show that classically this advantage can only exist if the sets of channels corresponding to the hypotheses are non-convex, and additionally assuming this convexity makes parallel strategies asymptotically optimal (\autoref{thm:classical_equality}). We leave the question open whether an adaptive advantage can exist in the quantum case when the sets of channels are convex. 
	\autoref{tab:main_results} gives an overview of what we are able to show regarding composite channel discrimination, and illustrates in which cases an adaptive advantage exists.
 
	As a consequence of our more general treatment which is not limited to the composite i.i.d. setting we also obtain a generalization of the composite state discrimination results of~\cite{berta_composite_2021} (\autoref{thm:general_states}). Note, however, that while we do not require provided states or channels to be identical, we still require them to be independent. Hence, our theorems do not aid in determining whether a generalized Stein's lemma holds in cases where the alternative hypothesis is given by a set of non-independent states, as conjectured in \cite{brandao_generalization_2010, berta_gap_2022}. 

     \newpage
    \subsection*{Summary of Main Results}

        \begin{center}
        \captionsetup{type=table}
        \begin{tabular}{c|c|c|c|c|c|c}
            \thead{Hypotheses} & \thead{Asymptotic \\ Parallel Exponent} & & \thead{Asymptotic\\ Adaptive\\ Exponent} & & \thead{Upper\\ Bound} & \thead{Shown in}\\\hline {\color{gray}
            \makecell{Quantum \\ Simple}} & {\color{gray} $D_{\mathrm{reg}}(\E\|\F)$} & {\color{gray}$=$} & {\color{gray}$D_{A}(\E\|\F)$} &  & & {\color{gray}\makecell{\cite{wang_resource_2019}\\\cite{ wilde_amortized_2020}\\\cite{fang_chain_2020}}}\\\hline
            \makecell{Classical\\ Composite \\ Convex Sets} & $ \displaystyle \max_{\nu} \min_{\substack{\E \in \S \\ \F \in \T}} D(\E(\nu)\|\F(\nu))$ & $=$ & $ \displaystyle \min_{\substack{\E \in \S \\ \F \in \T}} D(\E\|\F)$ & $=$ & $\displaystyle \min_{\substack{\E \in \S \\ \F \in \T}} D(\E\|\F)$ & Thm. \ref{thm:classical_equality}\\\hline
            \makecell{Classical\\ Composite \\ Finite Sets} & $ \displaystyle \max_{\nu} \min_{\substack{\E \in \S \\ \F \in \T}} D(\E(\nu)\|\F(\nu))$ & \makecell{$<$\\i.g.} & $?$ & \makecell{$\leq$} & $ \displaystyle \min_{\substack{\E \in \S \\ \F \in \T}}  D(\E\|\F)$ & \makecell{Prop. \ref{prop:classical_finite} \\ Ex. \ref{example:adaptive_advantage} \\ Prop. \ref{prop:adaptive_upper_bound}}\\\hline
            \makecell{Quantum\\ Composite} & $\displaystyle \lim_{n \to \infty} {1 \over n} \min_{\substack{\E_n \in \C(\S_n) \\ \F_n \in \C(\T_n)}} D(\E_n\|\F_n)$ & \makecell{$<$\\i.g.} & $?$ & \makecell{$<$\\i.g.} & $\displaystyle \min_{\substack{\E \in \S \\ \F \in \T}} D_A(\E\|\F)$ & \makecell{Thm. \ref{thm:parr} \\ Ex. \ref{example:adaptive_advantage} \\ Prop. \ref{prop:adaptive_upper_bound} \\
            Rem. \ref{remark:upper_bound_loose}}
        \end{tabular}
        \caption{Illustration of the relation between adaptive and parallel type II error exponents for various channel discrimination tasks. For the composite problems the task is to discriminate between two sets of channels $\S$ and $\T$ and the table also includes an upper bound based on the worst-case simple i.i.d. problem. 
        \enquote{Quantum Simple} refers to the quantum channel discrimination problem with simple hypotheses. With \enquote{Classical} we mean that all channels are classical, and \enquote{Convex Sets} or \enquote{Finite Sets} refers to whether the sets of channels $\S$ and $\T$ are convex or finite.
        Please see the respective theorems for a general formulation of the results and a precise definition of the quantities involved; $\C$ denotes the convex hull. We write i.g. to denote that these inequalities will be strict in general, although there exist specific examples where equality holds.}
        \label{tab:main_results}
        \end{center}
	
	\section{Mathematical Preliminaries}	 
	We write $\HS$ for a complex finite-dimensional Hilbert space, and $\BO$ for the set of linear operators acting on $\HS$. We write $\pos$ for the set of positive semi-definite operators acting on $\HS$. For $A, B \in \pos$, we further write $A \ll B$ if $\supp{A} \subseteq \supp{B}$  and $A \not \ll B$ if $\supp{A} \not \subseteq \supp{B}$.  Let $\DM$ denote the set of density matrices, i.e., the set of positive semi-definite operators with trace one. A quantum channel (in this paper usually denoted as $\E$ or $\F$) is a completely positive trace preserving map between density operators. We will label different quantum systems by capital Roman letters ($A$, $B$, $C$, etc.) and often use these letters interchangeably with the corresponding Hilbert space or set of density matrices (i.e., we write $\rho \in \DM[A]$ instead of $\rho \in \DM[\HS_A]$ and $\E: A \to B$ instead of $\E: \DM[\HS_A] \to \DM[\HS_B]$). We will also concatenate these letters to mean tensor products of systems, i.e.~we will write $\rho \in \DM[RA]$ for $\rho \in \DM[\HS_R \otimes \HS_A]$. We write $\cptp$ for the set of all completely positive trace preserving maps from $\DM[\HS_A]$ to $\DM[\HS_B]$. Throughout the paper we will write $\X$ and $\Y$ for classical systems. A classical state $\rho \in \DM[\X]$ is then diagonal in the computational basis, and we write $\cptp[\X \to \Y]$ for the set of classical channels. 
	
	For any subset $A$ of a vector space, we will write $\C(A)$ for the convex hull of $A$. We write $\log$ for the logarithm to the base two.

	\subsection{Measurements and POVMs}\label{sec:povms}
	Throughout this paper we will treat measurements and POVMs as quantum-classical channels, i.e.~we associate a POVM specified through the operators $\{M_i\}_{i=1}^{n} \subset \BO[\HS]$ with the quantum-classical channel
	
	\begin{equation}
		\M : \DM[\HS] \to \DM[\mathbb{C}^{n}] \qquad      \rho \mapsto \sum_{i = 1}^{n} \Tr(\rho M_i) \ketbra{i}\,.
	\end{equation}

	\subsection{Quantum Information Measures}
	
	For $\rho \in \DM$ and $\sigma \in \pos$ the (Umegaki) quantum relative entropy is defined as~\cite{umegaki_conditional_1962}
	\begin{equation}
		D(\rho\|\sigma) \coloneqq \Tr(\rho(\log \rho - \log \sigma)),
	\end{equation}
	if $\rho \ll \sigma$ and 
	$D(\rho\|\sigma) \coloneqq \infty$ if $\rho \not \ll \sigma$. One of its most important properties is the data-processing inequality~\cite{lindblad_completely_1975}, which states that for every quantum channel $\E$:
	\begin{equation}
		D(\rho\|\sigma) \geq D(\E(\rho)\|\E(\sigma)) \,.
	\end{equation}
	A self-contained proof can be found e.g.~in~\cite{khatri_principles_2020}.
	More generally, we call a function of $\rho$ and $\sigma$ a divergence if it satisfies the data-processing inequality.

	We can also define the measured relative entropy as the maximal classical relative entropy when measuring both states with some POVM. Specifically
	
	\begin{equation}
		D_M(\rho\|\sigma) \coloneqq \sup_{\M \text{ POVM}} D(\M(\rho)\|\M(\sigma))
	\end{equation}
	where $\M$ is a POVM (with arbitrarily many outcomes) interpreted as a quantum-classical channel as outlined above.
 
	For $\rho \in \DM$ and $\sigma \in \pos$, define the quantum max-divergence (or the max-relative entropy) as \cite{datta_min-_2009}
	\begin{equation}
		D_{\max}(\rho\|\sigma) \coloneqq \log \inf\Set{\lambda \in \reals | \rho \leq \lambda \sigma}\,.
	\end{equation}
	The quantum max-divergence also satisfies the data-processing inequality \cite{datta_min-_2009}.

	We will also frequently use the hypothesis testing relative entropy, which for $\rho, \sigma \in \DM[\HS]$ is defined as follows~\cite{wang_one-shot_2012}
	\begin{equation}
		D_H^{\varepsilon}(\rho\|\sigma) \coloneqq - \log \qty[ \min_{\substack{0 \leq M \leq \IdentityMatrix_{\HS} \\ \Tr(M \rho) \geq 1- \varepsilon}} \Tr(\sigma M)]\,.
	\end{equation}

	\subsubsection{Channel Divergences}
	
	For every given divergence $\mathbf{D}$ for states, one can define an associated channel divergence~\cite{leditzky_approaches_2018} by performing a (stabilized) maximization over all input states, i.e.~with $\E, \F: A \to B$ being quantum channels
	\begin{equation}
		\mathbf{D}(\E\|\F) \coloneqq \sup_{\rho_{RA} \in \DM[R \otimes A]}\mathbf{D}\big((\id_R \otimes \E)(\rho)\|(\id_R \otimes \F)(\rho)\big).
	\end{equation}
	Since $\mathbf{D}$ satisfies the data-processing inequality by definition, the supremum can be restricted to pure states such that the reference system $R$ is isomorphic to the channel input system $A$ (this is shown below as \autoref{lem:restrict_reference_size}). 

    For the Umegaki relative entropy $D$, we can also define the regularized and amortized \cite{wilde_amortized_2020} channel divergences as
    \begin{align}
        D_{\mathrm{reg}}(\E\|\F) &\coloneqq \lim_{n \to \infty}{1 \over n} D(\E^{\otimes n}\|\F^{\otimes n}), \\
        D_{A}(\E\|\F) &\coloneqq \sup_{\substack{\rho, \sigma \in \DM[RA] \\ R \text{ arbitrary}}}\qty[D(\E(\rho)\|\F(\sigma)) - D(\rho\|\sigma)].
    \end{align}
    Note that for the amortized divergence, there is no known way in which the size of the reference system can be restricted. 
	
	\section{Composite State Discrimination}\label{sec:composite-state-discrimination}
	In simple quantum state discrimination, given $n$ identical copies of an unknown state which is promised to be either $\rho$ or $\sigma$, the task is to decide which of the two options it is. In composite quantum state discrimination, we are only promised that the states are all from one of two sets $S$ or $T$, and the task is to decide which set they come from (but not to further identify which state exactly was provided). Since there are now multiple states for each hypothesis, there are multiple possible scenarios how the $n$ input states one receives are related: We could still be given $n$ identical copies of a state, or alternatively, we could be given $n$ completely different states but all from the same set $S$ or $T$, or something in between, where the states are non-identical but still related. We would like to cover all these different scenarios in our analysis, and hence we will describe composite hypotheses as sequences of sets $S_n$ which include all the possible combinations of $n$ states we could get. We will make some small assumptions on these sets:
	
	\begin{definition}\label{def:composite_state_hypothesis}
		For the purpose of this work, a composite quantum state hypothesis (in the asymptotic setting) is a sequence of sets of states $$\bm{S} = (S_n \subset \DM[\HS^{\otimes n}])_n$$ such that
		\begin{enumerate}
			\item Each set $S_n$ is topologically closed.
			\item Each element $\rho_n \in S_n$ is a tensor product of states $\rho_n = \rho^{(1)} \otimes \ldots \otimes \rho^{(n)}$, with each $\rho^{(i)} \in \DM[\HS]$ for $i = 1, ..., n$. 
			\item The sets $S_n$ are closed under tracing out any subsystem, i.e.~for any $i = 1, ..., n$ and $\rho_n \in \S_n$ we have that $\Tr_i(\rho_n) \in \S_{n-1}$, where $\Tr_i$ denotes the partial trace over the $i^{th}$ subsystem.
			\item Each set $S_n$ is closed under permutation of the $n$ subsystems, i.e.~for any permutation $\pi \in \Sn$ and associated canonical unitary representation $P_{\HS}(\pi)$, we have for all $\rho_n \in S_n$: $P_{\HS}(\pi) \rho_n P_{\HS}(\pi)^\dagger \in \S_n$. 
		\end{enumerate}
		
	\end{definition}

	Interesting examples of this include:
	\begin{enumerate}
		\item The composite i.i.d.\ case: We have two sets $S$, $T \subset \DM$, and are given $n$ identical copies of an element from $S$ if the null hypothesis is true, and $n$ identical copies of an element from $T$ if the alternate hypothesis is true. This corresponds to:
		\begin{align}
			S_n &\coloneqq \Set{\rho^{\otimes n}| \rho \in S}, \\
			T_n &\coloneqq \Set{\sigma^{\otimes n}| \sigma \in T}.
		\end{align}
		\item The arbitrarily varying case: This is similar to the composite i.i.d. case, but we are not given $n$ identical copies, but $n$ (potentially different) elements from $S$ or $T$. This corresponds to:
		\begin{align}
			S_n &\coloneqq \Set{\rho_1 \otimes \ldots \otimes \rho_n| \rho_1, \ldots, \rho_n \in S}, \\
			T_n &\coloneqq \Set{\sigma_1 \otimes \ldots \otimes \sigma_n| \sigma_1, \ldots,\sigma_n \in T}.
		\end{align}
		\item The slightly-varying case: This is an example of a scenario that lies in between the arbitrarily varying case (where there is no correlation between the samples, except for them all being in the same set) and the composite i.i.d. case (where there is maximal correlation between the samples, as they are all identical). For any given $\varepsilon \in [0,1]$ (which might depend on $n$) and any distance function $d : \DM \cross \DM \to [0,1]$ (e.g.~ trace distance or purified distance) set
		\begin{align}
			S_n &\coloneqq \Set{\rho_1 \otimes \ldots \otimes \rho_n| \rho_1, \ldots, \rho_n \in S, \quad d(\rho_i, \rho_j) \leq \varepsilon \; \forall i,j}, \\
			T_n &\coloneqq \Set{\sigma_1 \otimes \ldots \otimes, \sigma_n| \sigma_1, \ldots, \sigma_n \in T,\quad d(\sigma_i, \sigma_j) \leq \varepsilon \; \forall i,j}.
		\end{align}
		\item The simple i.i.d.\ case: The simple i.i.d.\ case can be seen as a special case of the above where $S$ and $T$ each contain one element. 
	\end{enumerate}

	\begin{lemma}\label{lem:partial_measurement_conditioning}
		If $\bm{S}$ is a quantum state hypothesis, then performing a measurement on any (joint) $k$ subsystems of a state $\rho_n \in S_n$, and conditioning on the measurement result, yields a state $\rho_{n-k}$ on the remaining subsystems that is an element of $S_{n-k}$. 
		Precisely, let $k \in \{1, ..., n\}$, $\rho_n \in S_n \subset \DM[\HS^{\otimes n}]$ and $0 \leq M \leq \IdentityMatrix \in \BO[\HS^{\otimes k}]$. If
		\begin{equation}
			\omega_{n-k} \coloneqq \Tr_{1, ..., k}\qty[ (M \otimes \IdentityMatrix_{\HS_{k + 1}} \otimes ... \otimes \IdentityMatrix_{\HS_n}) \rho_n]
		\end{equation}
		then either $\tr(\omega_{n-k}) = 0$ or $\omega_{n-k}/\tr(\omega_{n-k}) \in S_{n-k}$.
	\end{lemma}
	\begin{proof}
		This follows immediately from the fact that each element $\rho_n \in S_n$ is a tensor product, and that removing an element in the tensor product gives an element in $S_{n-1}$. 
	\end{proof}

    \begin{remark}
        Note that while we state the results below with the assumptions of \autoref{def:composite_state_hypothesis}, we could replace the required tensor product structure (point 2 in \autoref{def:composite_state_hypothesis}) with the statement of \autoref{lem:partial_measurement_conditioning}. While this might be more general, we find the assumptions of \autoref{def:composite_state_hypothesis} to be more natural. Similarly, later on when talking about hypotheses in the context of channel discrimination, we could replace the tensor product structure for channels (point 2 in \autoref{def:composite_channel_hypothesis}) with the statement of \autoref{lem:partial_measurement_conditioning} for any tensor product input state. 
    \end{remark}
 	
	For our discrimination problem, given an $n$ and an unknown state in $\DM[\HS^{\otimes n}]$ we will still want to perform a binary POVM (fully specified by one of its elements, which we write as $M$) to decide between the two hypotheses. In the end we want to avoid making an error, i.e.~claiming that our state comes from $S_n$ when it actually comes from $T_n$ and vice-versa, and these two error probabilities are again what we call type-I and type-II error now in this setting (see below for a formal definition). If we settle on a measurement $M$, the probability of making an error might still depend on which particular state from either $S_n$ or $T_n$ we actually end up getting. Here, we will be focussing on minimizing the worst case errors, i.e.~we want to choose measurements which minimize the error uniformly over all states from $S_n$ and $T_n$. More formally, we define the type I and type II error probabilities (also again called type I and type II errors) as:
	\begin{align}\label{eq:def_errors_states}
		\alpha(M, S_n) &= \sup_{\rho \in S_n} \Tr((\IdentityMatrix - M) \rho) \\
		\beta(M, T_n) &= \sup_{\sigma \in T_n} \Tr(M \sigma)\,.
	\end{align}
	We will be focussing on the asymmetric setting, where we want to minimize the type II error, $\beta$, under the constraint that the type I error $\alpha$ is below a certain threshold. The main quantity of interest is then the negative logarithm of this minimal type II error under the type I error constraint, which we also call the hypothesis testing relative entropy of the two sets $S_n$ and $T_n$:
	\begin{equation}
		D_H^{\varepsilon}(S_n\|T_n) \coloneqq - \inf_{\substack{0 \leq M \leq \IdentityMatrix \\ \alpha(M, S_n) \leq \varepsilon}}  \log \beta(M, T_n) \,.
	\end{equation}
	
	As the expression in \eqref{eq:def_errors_states} is linear in $\rho$, it is easy to see that 
	\begin{equation}
		\alpha(M, \C(S_n)) = \alpha(M, S_n)
	\end{equation}
	(where $\C$ is the convex hull)
	as the supremum will be achieved at an extremal point, and the same holds also for $\beta$. 
	
	Hence,
	\begin{equation}
		D_H^{\varepsilon}(S_n\|T_n) = D_H^{\varepsilon}(S_n\|\C(T_n)) = D_H^{\varepsilon}(\C(S_n)\|T_n) = D_H^{\varepsilon}(\C(S_n)\|\C(T_n))
	\end{equation}
	and hence the discrimination task considered is equivalent to discriminating between these convex hulls of the sets $S_n$ and $T_n$.
	We are interested in the quantum Stein exponent for this discrimination task, i.e.~the optimal exponential decay rate of the type II error in the limit $n \to \infty$ (which corresponds to infinitely many states being provided). Morally, the expression we are interested in is
	\begin{equation} \label{eq:composite_state_target_expr}
		\lim_{\varepsilon \to 0} {\essentiallylim_{n \to \infty}} \frac{1}{n} D^{\varepsilon}_H(S_n\|T_n)\,.
	\end{equation}
	However, due to the optimization over the elements from the sets $S_n$ and $T_n$, for any fixed $ε > 0$, $D^{\varepsilon}_H(S_n\|T_n)$ need not be super-additive in $n$, and hence we are not aware of any way to show that the limit $n \to \infty$ exists in these expressions. However, in many cases we are able to show that after additionally taking the limit $ε \to 0$ outside, the final expression does not depend on whether one takes a $\liminf$ or a $\limsup$ inside. Hence, we introduce the following notation
	
	\begin{definition}\label{def:liminfsup}
		For any function $f(n, ε): \naturals \cross (0,1) \to \reals$, and $A \in \reals \cup \{- \infty, \infty\}$, we write
		\begin{equation}
			\lim_{ε \to 0} \liminfsup_{n \to \infty} f(n, ε) = A
		\end{equation}
		if
		\begin{equation}
				\lim_{ε \to 0} \liminf_{n \to \infty} f(n, ε)  = A
		\end{equation}
		and 
		\begin{equation}
			\lim_{ε \to 0} \limsup_{n \to \infty} f(n, ε) = A\,.
		\end{equation}
	\end{definition}
		
	To come back to \eqref{eq:composite_state_target_expr}, in~\cite{berta_composite_2021} this problem was studied specifically in the composite i.i.d.~case, i.e.~with $S_n = \Set{\rho^{\otimes n} | \rho \in S}$, $T_n = \Set{\sigma^{\otimes n}| \sigma \in T}$, where $S, T \subset \DM$ were also assumed to be closed and convex, leading to:
	\begin{theorem}[\cite{berta_composite_2021}] \label{thm:berta_composite}
		Let $S, T$ be closed and convex, and define for all $n$: $S_n \coloneqq \Set{\rho^{\otimes n} | \rho \in S}$, $T_n \coloneqq \Set{\sigma^{\otimes n} | \sigma \in T}$. Then
		
		\begin{equation}\label{eq:berta_thm_main}
			\lim_{\varepsilon \to 0} \liminfsup_{n \to \infty} \frac{1}{n} D^{\varepsilon}_H(S_n\|T_n) = \lim_{n \to \infty} \frac{1}{n} \inf_{\substack{\rho_n \in \S_n \\ \sigma_n \in \C(T_n)}} D(\rho_n \|\sigma_n) \,,
		\end{equation}
		where $\liminfsup$ can be $\liminf$ or $\limsup$ (see \autoref{def:liminfsup}), and one can find cases where this is strictly smaller than
		\begin{equation}\label{eq:berta_thm_collapse}
			\inf_{\substack{\rho \in S \\ \sigma \in T}} D(\rho\|\sigma)\,.
		\end{equation}
	\end{theorem}
	Remember that $\C$ stands for the convex hull, and it is precisely this convex hull in the infimum on the right-hand side of \eqref{eq:berta_thm_main} which prevents the regularization from collapsing, as the elements of $S_n$ and $T_n$ are tensor products, and the relative entropy is additive. Without this convex hull, the exponent \eqref{eq:berta_thm_main} would be exactly equal to the single-letter expression \eqref{eq:berta_thm_collapse}, which we call the worst-case i.i.d. exponent, as it is equal to the exponent of the worst-case simple i.i.d. problem. Intuitively, one pays for the compositeness by having to include convex combinations in the Stein exponent, and this makes the discrimination problem strictly harder in some cases. 

	As a consequence of our channel discrimination result further below (\autoref{thm:parr}), we will arrive at this following generalization of \autoref{thm:berta_composite}:
	
	\begin{theorem}\label{thm:general_states}
		Let $\bm{S} = (S_n)_n, \bm{T} = (T_n)_n$ be two composite quantum state hypotheses. Then
		\begin{equation}
			\lim_{\varepsilon \to 0} \liminfsup_{n \to \infty} \frac{1}{n} D_H^{\varepsilon}(S_n\|T_n) = \lim_{n \to \infty} \frac{1}{n} \min_{\substack{\rho_n \in \C(\S_n) \\ \sigma_n \in \C(T_n)}} D(\rho_n \|\sigma_n) \,,
		\end{equation}
		where $\liminfsup$ can be $\liminf$ or $\limsup$ (see \autoref{def:liminfsup}).
		Furthermore, if each $S_n$ lies in the intersection of $\DM[\HS^{\otimes n}]$ with a linear subspace of $\BO[\HS^{\otimes n}]$ with dimension polynomial in $n$ (this holds for example in the composite i.i.d. case, where each $\rho_n \in S_n$ is permutation invariant), then we can remove the first convex hull to get
		\begin{equation}
			\lim_{\varepsilon \to 0} \liminfsup_{n \to \infty} \frac{1}{n} D^{\varepsilon}_H(S_n\|T_n) = \lim_{n \to \infty} \frac{1}{n} \min_{\substack{\rho_n \in \S_n \\ \sigma_n \in \C(T_n)}} D(\rho_n \|\sigma_n) \,.
		\end{equation}
		
	\end{theorem}
	\begin{proof}
		This follows as a special case of \autoref{thm:parr} below. 
	\end{proof}
	\begin{remark}
		Our \autoref{thm:general_states} generalizes the previous \autoref{thm:berta_composite} in multiple ways: Already in the composite i.i.d. setting it no longer requires the sets $\S$ and $\T$ to be convex. Additionally our theorem also includes all the non-i.i.d. cases such as the arbitrarily (or slightly varying) cases defined above. 
	\end{remark}
	
	\subsection{Classical Adversarial Hypothesis Testing}
	Similar to~\cite{brandao_adversarial_2020, berta_composite_2021} our results are based on a reduction to a classical problem, the one of adversarial hypothesis testing. The following is a brief recapitulation of the treatment of adversarial hypothesis testing in~\cite{brandao_adversarial_2020}. 
	Let $P, Q \subset \reals^{\Omega}$ (for a finite domain $\Omega$) be two sets of probability distributions. In the typical composite i.i.d.~setting, we are presented with $n$ samples from a distribution in $P$ or $Q$ and have to make a decision which set the distribution comes from. In the adversarial setting, the adversary is allowed to change the distribution within $P$ or $Q$ for each sample, and he can make this change based on the samples we observed previously. Note that while the adversary has access to the previous samples, he can only select a probability distribution $p \in P$ or $q \in Q$ (depending on which hypothesis is true) for the next sample, but he cannot select the sample outcome itself. The adversary is fully specified by two sets of functions $\hat{p}_k : \Omega^{k-1} \to P$ and $\hat{q}_k: \Omega^{k-1} \to Q$, which for each $k$ specify how the adversary picks the next probability distribution based on the previous $k-1$ sample outcomes. The two hypotheses then correspond to whether the adversary uses $\hat{p}_k$, and hence always chooses a probability distribution in $P$, or $\hat{q}_k$ and always chooses a probability distribution in $Q$. If the null hypothesis is true (i.e. the adversary uses $\hat{p}_k$), the probability of a sample string $\mathbf{x} \in \Omega^n$ is then given by
	\begin{equation}
		\hat{p}(\mathbf{x}) \coloneqq \prod_{k = 1}^{n} \hat{p}_k(x_1, \ldots, x_{k-1})(x_k)\, ,
	\end{equation} and we define $\hat{q}(\mathbf{x})$ in a similar manner.
	For any decision region $A_n \subset \Omega^n$, the type I and type II errors are then going to be the worst-case errors over all adversarial strategies. We define the corresponding $n$-shot error exponent as	 
	\begin{equation}
		D_{\mathrm{adv}, n}^{\varepsilon}(P\|Q) = - \log \inf\Set{\sup_{\hat{q}} \hat{q}(A_n) | A_n \subset \Omega^n, \; \sup_{\hat{p}} \hat{p}(A_n^{c}) \leq \varepsilon}
	\end{equation}
	The key statement of~\cite{brandao_adversarial_2020} is that if the sets $P$ and $Q$ are closed and convex, adversarial hypothesis testing is asymptotically no harder than the worst-case i.i.d. setting, specifically:
	
	\begin{theorem}[{\cite[Theorem 2]{brandao_adversarial_2020}}]\label{thm:adversarial}
		Let $\Omega$ be a finite domain and $P, Q \subset \reals^{\Omega}$ be two closed, convex sets of probability distributions. Then, for any $\varepsilon \in (0,1)$:
		\begin{equation}
			\lim_{n \to \infty} {1 \over n} D^ε_{\mathrm{adv}, n}(P\|Q)  = \min_{p \in P, q \in Q} D(p \|q)\,.
		\end{equation}		
	\end{theorem}
    Note that since we are taking the supremum over all adversaries, by picking an adversary that deterministically picks states in a certain sequence, this result implies that also any composite problem is classically asymptotically equally as hard as the worst-case i.i.d. problem (it is also easy to see that the composite problem cannot be simpler than the worst-case i.i.d. problem).

	\section{Composite Channel Discrimination}
		The task of composite channel discrimination is very similar in nature to composite state discrimination, but also considerably harder to study: Given an unknown quantum channel as a black box and the side information that it comes from two sets of possible channels, the task is again to determine the set (but not necessarily the exact identity) of the channel.  
		The additional complexity here comes from the fact that, on top of finding the best measurement to perform on the output of the channel, we also have to figure out which quantum states to send as inputs to the channel. 
  
        If we are given access to the black box multiple times (or say we are given multiple black-boxes) the problem becomes considerably more interesting, as the channel inputs could be chosen based on previous channel outputs. Say we are given access to $n$ black-boxes (we will allow for the case where not all black-boxes are identical and will specify further below what scenarios exactly we consider, intuitively though the scenario is always that we want to distinguish $n$ black-boxes from one set to $n$ black-boxes from another set). There are now different strategies (sometimes also called protocols) in which we could set up our decision experiment -- the so-called {\em{parallel}} and {\em{adaptive}} strategies.

		In a parallel strategy one prepares a joint state, usually entangled between the input systems of all the $n$ channels and an additional reference (or memory) system. This state is then fed as input to all the $n$ channels at once (with the state of the reference system being left undisturbed). Finally, a binary positive operator-valued measure (POVM) is performed on the joint state at the output of the channels and the reference system in order to arrive at a decision.
		In an adaptive strategy, on the other hand, one prepares a state of the input system of a single channel (again usually entangled with a reference system) which is fed into the first channel, with the state of the reference system being left undisturbed. The input to the next use of the channel is then chosen depending on the output of the first channel and the state of our reference system. This is done, most generally, by subjecting the latter to an arbitrary quantum operation (or channel), which we call a preparation operation. This step is repeated for each successive black-box channel until all the $n$ black-boxes have been used.
		Then a binary POVM is performed on the joint state of the output of the last use of the channel and the reference system. See \autoref{fig:adaptive_strategy} for a depiction of an adaptive strategy. Adaptive strategies are also sometimes called sequential, which is, however, not to be confused with the setting of sequential hypothesis testing~\cite{martinez_vargas_quantum_2021, li_optimal_2022, li_sequential_2022}, where samples (i.e.~states or channels) can be requested one by one. 
		
		One particularly interesting question is whether and to what degree adaptive strategies give an advantage over parallel ones. Note that any parallel strategy can be written as an adaptive strategy by taking all but one channel input as part of the reference system, and then choosing each preparation operation such that it extracts the next part of the joint input state for the next channel use and replaces it by the output of the previous channel use. However, the converse is not true, and so adaptive strategies are more general. Parallel strategies are conceptually a lot simpler than adaptive ones -- aside from the measurement, everything is specified just by the joint input state -- in contrast to adaptive strategies, in which after each channel use we can perform an arbitrary quantum operation to prepare the input to the next use of the channel. It is thus interesting to determine to what degree parallel strategies can still be optimal.
		
		This problem has been studied extensively for channel discrimination with simple hypotheses, where it is known that in certain cases adaptive strategies can give an advantage over parallel ones. In~\cite{harrow_adaptive_2010} the authors constructed an example in which an adaptive strategy with only two channel uses could be used to discriminate two channels with certainty, which is shown not to be possible with a parallel strategy, even if arbitrarily many channel uses are allowed. Asymptotically, however, it was shown that in the simple binary asymmetric case adaptive and parallel strategies are equivalent~\cite{wang_resource_2019, fang_chain_2020, wilde_amortized_2020}. We will show below that this fails to stay the case with composite hypotheses, already classically. 
		
		Specifically, in this section we will study the following:
		1. We start with a treatment of parallel channel discrimination strategies, where we provide matching achievability and converse bounds for the Stein exponent in terms of a regularized expression (\autoref{thm:parr}), in analogy to what has previously been shown~\cite{berta_composite_2021} for state discrimination (i.e.~\autoref{thm:berta_composite}).  
		2. We prove an upper bound on the Stein exponent for adaptive strategies (\autoref{prop:adaptive_upper_bound}), where we show that this upper bound can sometimes but not always be achieved, and can also be larger than the parallel exponent (\autoref{example:adaptive_advantage}), hence demonstrating that adaptive strategies can sometimes be advantageous (we show this even classically). 
		3. We show that classically, under an additional convexity assumption which was not satisfied in the previous example, parallel and adaptive strategies are asymptotically equivalent in the asymmetric composite setting, and the Stein exponent is given by a single-letter entropic formula (\autoref{thm:classical_equality}). 
        4. We further show classically, and in some further restricted setting, that if we replace the convexity assumption with a finiteness assumption, we can still get a single-letter entropic expression for the Stein exponent for parallel strategies (\autoref{prop:classical_finite}).

	Following the above discussion for composite state discrimination, we want to apply a similar level of generality to discriminating channels, where we want to allow the $n$ black-boxes not be identical. Hence, in analogy with \autoref{def:composite_state_hypothesis} we will work with general hypotheses satisfying the following conditions:
	
	\begin{definition}\label{def:composite_channel_hypothesis}
		For the purpose of this work, a composite quantum channel hypothesis (in the asymptotic setting) is a sequence of sets of channels $$\bm{\S} = (\S_n \subset \cptp[A^n \to B^n])_n$$ such that
		\begin{enumerate}
			\item Each set $\S_n$ is topologically closed.
			\item Each element $\E_n \in \S_n$ is a tensor product of channels $\E_n = \E^{(1)} \otimes \ldots \otimes \E^{(n)}$, with $\E^{(i)} \in \cptp[A \to B]$, for $i = 1, ..., n$. 
			\item For every $\E_n = \E^{(1)} \otimes \ldots \otimes \E^{(n)} \in \S_n$, removing any element in the tensor product (i.e. discarding one of the $n$ provided channels) yields an element in $\S_{n-1}$.
			\item Each set $\S_n$ is closed under permuting the $n$ subsystems of the input and output systems of a channel, i.e.~for any permutation $\pi \in \Sn$ and associated canonical unitary representations $P_A(\pi)$ and $P_B(\pi)$ on $A^n$ and $B^n$, we have for all $\E_n \in \S_n$ that also the permuted channel $\rho \mapsto P_B(\pi) \E_n(P_A(\pi) \rho P_A(\pi)^\dagger) P_B(\pi)^\dagger$ is an element of $\S_n$. 
		\end{enumerate}
		
	\end{definition}

 One can then define the same scenarios, such as the composite i.i.d. setting, the arbitrarily varying setting, and slightly varying settings, as we did for composite state discrimination (below \autoref{def:composite_state_hypothesis}) in a completely analogous way for composite channel discrimination. Note that since the $n$ channels one receives in the composite hypothesis testing problem need not all be equivalent, one might think that (in particular in an adaptive strategy) one might want to order the channels in a certain way, however it is not hard to see that this does not give any advantage, and so we can restrict to strategies that just take the channels in the order they are given. To see this, note that we assumed the sets of channels to be closed under permutations and we are looking at worst-case error probabilities. Hence, for every reordering one would perform for a given sequence of channels, there exists another sequence in the set that inverts this reordering, and hence in the worst-case one cannot gain anything.

	\subsection{The Parallel Case}
	Given a set of channels $\mathcal{A}$ and an input state $\nu \in \DM[RA]$ (where $R$ could be any system, possibly also just trivial), we define the set of all output states as
	\begin{equation}
		\mathcal{A}[\nu] \coloneqq \Set{(\id_R \otimes \E)(\nu) | \E \in \mathcal{A}}\,.
	\end{equation}
	Since we want to be looking at worst-case errors again (as introduced for composite state discrimination in \Cref{sec:composite-state-discrimination}), we will be looking for the best input state $\nu_n$ and measurement $M$, such that \emph{for all} $\E_n \in \S_n$ the error of claiming it coming from $\T_n$ (i.e.~the type I error) stays below some threshold $\varepsilon$ and we otherwise minimize the worst case type II error, i.e.~we want to make sure that the probability of claiming an element $\F_n \in \T_n$ to be from $\S_n$ is as low as possible uniformly over all $\F_n \in \T_n$. 
	Given a joint input state $\nu$, the parallel channel discrimination problem turns into a state discrimination problem, and so we define the following type II error exponent for any $\S_n$ and $\T_n$ which satisfy the properties of \autoref{def:composite_channel_hypothesis}:
	\begin{align}\label{eq:parr_exponent_dh}
		D_H^{\varepsilon}(\S_n\|\T_n) &\coloneqq \sup_{\nu \in \DM[R A]} D^{\varepsilon}_H(\S_n[\nu]\|\T_n[\nu]) = \sup_{\nu \in \DM[R A]} \sup_{\substack{0 \leq M \leq \IdentityMatrix \\ \alpha(M, \S_n[\nu]) \leq \varepsilon}} (- \log \beta(M, \T_n[\nu])) \\ 
		e_P(n, \varepsilon, \S_n, \T_n) &\coloneqq \frac{1}{n} D_H^{\varepsilon}(\S_n \|\T_n) \,.
	\end{align}	
	It is easy to see that $\C(\mathcal{A}[\nu]) = \C(\mathcal{A})[\nu]$, and hence (as above) for any two sets of channels $\S, \T$:
	\begin{equation}\label{eq:add_convex_hull_channels}
		D_H^{\varepsilon}(\S\|\T)  = D_H^{\varepsilon}(\S\|\C(\T)) = D_H^{\varepsilon}(\C(\S)\|\T) = D_H^{\varepsilon}(\C(\S)\|\C(\T))\, .
	\end{equation}
	Our main theorem of this section is the following:
	\begin{theorem}\label{thm:parr}
		Let $\bm{\S} = (\S_n)_n, \bm{\T} = (\T_n)_n$ be two composite quantum channel hypotheses (as defined in \autoref{def:composite_channel_hypothesis}). Then, the quantum Stein exponent of discriminating these two hypotheses with a parallel strategy is given by:
		\begin{equation}
			\lim_{\varepsilon \to 0} \liminfsup_{n \to \infty} \frac{1}{n} D_H^{\varepsilon}(\S_n\|\T_n) = \lim_{n \to \infty} \frac{1}{n}  \min_{\substack{\E_n \in \C(\S_n)\\ \F_n \in \conv(\T_n)}} \max_{\nu \in \DM[R \otimes A^{\otimes n}]} D(\E_n(\nu) \| \F_n(\nu))\,,
		\end{equation}
		where $\liminfsup$ can be $\liminf$ or $\limsup$ (see \autoref{def:liminfsup}). Additionally, on the right-hand side the $\min$ and $\max$ can be exchanged, and one can choose the reference system $R$ to be isomorphic to $A^{\otimes n}$ for all $n$.      
		
		Furthermore, if each $\S_n$ lies in the intersection of $\cptp[A^n \to B^n]$ with a linear subspace, with dimension polynomial in $n$, of the space of linear maps $A^n \to B^n$  (this is for example the case in the composite i.i.d. setting, where all the $\E_n \in \S_n$ are permutation covariant), we can also remove one convex hull:
		\begin{equation}
			\lim_{\varepsilon \to 0} \liminfsup_{n \to \infty} \frac{1}{n} D_H^{\varepsilon}(\S_n\|\T_n) = \lim_{n \to \infty} \frac{1}{n}  \max_{\nu \in \DM[R \otimes A^{\otimes n}]} \min_{\substack{\E_n \in \S_n\\ \F_n \in \conv(\T_n)}}  D(\E_n(\nu) \| \F_n(\nu))\,,
		\end{equation}
		where we however cannot say whether $\min$ and $\max$ can be exchanged.     
	\end{theorem}
	
	\begin{proof}
		This proof is very much inspired by the results for composite state discrimination from~\cite[Theroem 1.1]{berta_composite_2021} and~\cite[Theorem 16]{brandao_adversarial_2020}.
		\paragraph{Achievability}
		For the achievability part, let $\varepsilon \in (0,1)$, fix an integer $k$, and let $\nu_k \in \DM[RA^k]$ be an input state, where $R$ is isomorphic to $A^k$. Additionally, let $\M_k$ be a POVM measurement on $R B^k$ (where we interpret $\M_k$ as a quantum-classical channel that maps to the probability distribution of measurement outcomes). Define the two sets of classical probability distributions $P \coloneqq \Set{\M_k(\E_k(\nu_k))|\E_k \in \S_k}$ and $Q \coloneqq \Set{\M_k(\F_k(\nu_k)) |\F_k \in \T_k}$. The operational procedure is now to take an unknown channel from either $\S_{nk}$ or $\T_{nk}$, feed it with the input state $\nu_k^{\otimes n}$ and apply the measurement $\M_k^{\otimes n}$ to the outcome. Crucially, due to the assumed structure of the $(\S_n)_n$ and $(\T_n)_n$ (as specified in \autoref{def:composite_channel_hypothesis}), the measurement result of each of the individual $n$ POVM measurements will be distributed according to a $p \in P$ or $q \in Q$. Hence, the overall structure of classical outcomes can be seen as an instance of adversarial hypothesis testing with a particular adversary\footnote{In fact, this problem can also be seen to be at most as hard as a composite hypothesis testing task in the arbitrarily varying case, and a similar statement as \autoref{thm:adversarial} for this composite arbitrarily varying task would be sufficient for our purposes.}. For this classical problem, by \autoref{thm:adversarial}, the exponent 
		\begin{equation}
			\inf_{p \in P, q \in Q} D(p\|q)
		\end{equation}
		is asymptotically achievable as $n \to \infty$, which just means that
		\begin{equation}
			\liminf_{n \to \infty} \frac{1}{n} D^{\varepsilon}_H(\S_{nk}\|\T_{nk}) \geq \inf_{p \in P, q \in Q} D(p\|q) = \inf_{\substack{\rho_k \in \S_k[\nu_k]\\ \sigma_k \in \T_k[\nu_k]}} D(\M_k(\rho_k)\|\M_k(\sigma_k))\,,
		\end{equation}
		where dividing by $k$ yields:
		\begin{equation}
			\liminf_{n \to \infty} \frac{1}{nk} D^{\varepsilon}_H(\S_{nk}\|\T_{nk}) \geq {1 \over k} \inf_{\substack{\rho_k \in \S_k[\nu_k]\\ \sigma_k \in \T_k[\nu_k]}} D(\M_k(\rho_k)\|\M_k(\sigma_k))\,.
		\end{equation}
		
		Now, to obtain a procedure for discriminating $m$ channels where $m$ is not a multiple of $k$, we can just ignore at most $k-1$ channels so that we are left with a multiple of $k$ channels and then do the above. This yields a strategy to distinguish $\S_m$ and $\T_m$ for any $m$ and asymptotically the $k-1$ discarded channels do not matter, so we get:
		\begin{equation}
			\liminf_{m \to \infty} \frac{1}{m} D^{\varepsilon}_H(\S_{m}\|\T_{m}) \geq {1 \over k} \inf_{\substack{\rho_k \in \S_k[\nu_k]\\ \sigma_k \in \T_k[\nu_k]}} D(\M_k(\rho_k)\|\M_k(\sigma_k)) \geq \inf_{\substack{\rho_k \in \C(\S_k[\nu_k])\\ \sigma_k \in \C(\T_k[\nu_k])}} {1 \over k} D(\M_k(\rho_k)\|\M_k(\sigma_k)) \,,
		\end{equation}
		where we added convex hulls on the right-hand side (this just makes the infimum smaller). We can now take the supremum over all measurements $\M_k$ on the right-hand side, and by~\cite[Lemma 13]{brandao_adversarial_2020}, we can exchange this supremum with the already present infimum, to find
		\begin{equation}
			\liminf_{m \to \infty} \frac{1}{m} D^{\varepsilon}_H(\S_{m}\|\T_{m}) \geq  \inf_{\substack{\rho_k \in \C(\S_k[\nu_k])\\ \sigma_k \in \C(T_k[\nu_k])}} {1 \over k} D_M(\rho_k\|\sigma_k) \,.
		\end{equation}
		Note that~\cite[Lemma 13]{brandao_adversarial_2020} requires the infimum to be over a convex set, which is why we introduced convex hulls in the previous step. 
		Additionally, we now take the supremum over $\nu_k$ to find
		\begin{align}
			\liminf_{m \to \infty} \frac{1}{m} D^{\varepsilon}_H(\S_{m}\|\T_{m}) &\geq  \sup_{\nu_k \in \DM[R A^k]} \inf_{\substack{\rho_k \in \C(\S_k[\nu_k])\\ \sigma_k \in \C(T_k[\nu_k])}} {1 \over k} D_M(\rho_k\|\sigma_k) \\
			&= \sup_{\nu_k \in \DM[R A^k]} \inf_{\substack{\E_k \in \C(\S_k)\\ \F_k \in \C(T_k)}} {1 \over k} D_M(\E_k(\nu_k)\|\F_k(\nu_k)) \\
			&= \inf_{\substack{\E_k \in \C(\S_k)\\ \F_k \in \C(T_k)}} \sup_{\nu_k \in \DM[R A^k]} {1 \over k} D_M(\E_k(\nu_k)\|\F_k(\nu_k))
		\end{align}
		where the first equality is just a rewriting, and for the second equality we used that by \autoref{prop:main_exchange} (since the infimum is over convex sets, and $D_M$ satisfies the direct sum property) we can exchange infimum and supremum. 
		We take the $\limsup$ over $k$ to get
		\begin{equation}
			\liminf_{m \to \infty} \frac{1}{m} D^{\varepsilon}_H(\S_{m}\|\T_{m}) \geq  \limsup_{k \to \infty}\inf_{\substack{\E_k \in \C(\S_k)\\ \F_k \in \C(T_k)}} \sup_{\nu_k \in \DM[R A^k]} {1 \over k} D_M(\E_k(\nu_k)\|\F_k(\nu_k)) \,.
		\end{equation}
		Now, by \autoref{lem:inf_perm_covariant} the infimum is achieved for permutation covariant channels $\E_k$, $\F_k$, and by \autoref{lem:channel_div_perm_invariant} the supremum is achieved for a permutation invariant state (note that the channels $\E_k$ and $\F_k$ are of course also permutation covariant with regards to permutations within $R$, as they act with the identity on the reference system). Hence the state $\F_k(\nu_k)$ is permutation invariant, and thus by \autoref{lem:measured_asmptotically_equal} we get
		\begin{align}
			\liminf_{m \to \infty} \frac{1}{m} D^{\varepsilon}_H(\S_{m}\|\T_{m}) &\geq  \limsup_{k \to \infty}\min_{\substack{\E_k \in \C(\S_k)\\ \F_k \in \C(T_k)}} \max_{\nu_k \in \DM[R \otimes A^k]} {1 \over k} D(\E_k(\nu_k)\|\F_k(\nu_k))\\
			&= \limsup_{k \to \infty}\min_{\substack{\E_k \in \C(\S_k)\\ \F_k \in \C(T_k)}}  {1 \over k} D(\E_k\|\F_k) \,. \label{eq:channel_achievability_liminf}
		\end{align} 
		
		\paragraph{Converse}
		For the converse part, let $R \cong A$, and then note that by \autoref{lem:upper_bound_individual_mmt}:
		\begin{equation}\label{eq:channel_dhe_inf_outside}
			D_H^{\varepsilon}(\S_n\|\T_n) = \sup_{\nu_n \in \DM[R^n A^n]} D_H^{\varepsilon}(\S_n[\nu_n]\|T_n[\nu_n]) \leq \sup_{\nu_n \in \DM[R^n A^n]} \inf_{\substack{\rho_n \in \S_n[\nu_n] \\ \sigma_n \in 
					\T_n[\nu_n]}} D_H^{\varepsilon}(\rho_n\|\sigma_n)\,.
		\end{equation}
		By \autoref{lem:wang_renner_upper_bound} we have that for any two states $\rho, \sigma$:
		\begin{equation}
			D^{\varepsilon}_H(\rho\|\sigma) \leq \frac{1}{1 - \varepsilon} \qty( D(\rho\|\sigma) + h(\varepsilon))  \,.
		\end{equation}
		Thus,
		\begin{align}
			\lim_{\varepsilon \to 0} \liminf_{n \to \infty} \frac{1}{n} D_H^{\varepsilon}(\S_n\|\T_n) \alignedoverset{\eqref{eq:add_convex_hull_channels}}{=} \lim_{\varepsilon \to 0} \liminf_{n \to \infty} \frac{1}{n} D_H^{\varepsilon}(\C(\S_n)\|\C(\T_n)) \\
			&= \lim_{\varepsilon \to 0} \liminf_{n \to \infty} \frac{1}{n} \sup_{\nu_n \in \DM[R^n A^n]} D_H^{\varepsilon}(\C(\S_n)[\nu_n]\|\C(\T_n)[\nu_n])\\
			\alignedoverset{\eqref{eq:channel_dhe_inf_outside}}{\leq} 
			\lim_{\varepsilon \to 0} \liminf_{n \to \infty} \sup_{\nu_n \in \DM[R^n A^n]}  \min_{\substack{\rho_n \in \C(\S_n)[\nu_n]\\ \sigma_n \in \C(T_n)[\nu_n]}} \frac{1}{n} D_H^{\varepsilon}(\rho_n\|\sigma_n)\\
			&= \liminf_{n \to \infty} \sup_{\nu_n \in \DM[R^n A^n]} \min_{\substack{\E_n \in \C(\S_n)\\ \F_n \in \C(T_n)}}  {1 \over n} D(\E_n(\nu_n)\|\F_n(\nu_n)) \\
			\alignedoverset{\eqref{prop:main_exchange}}{=} \liminf_{n \to \infty}\min_{\substack{\E_n \in \C(\S_n)\\ \F_n \in \C(T_n)}}  {1 \over n} D(\E_n\|\F_n) \label{eq:channel_converse_liminf}
		\end{align}
		where the optimizations are achieved by the same argument as above. 
		Equivalently, one finds the same with $\liminf$ replaced with $\limsup$:
		\begin{equation}
			\lim_{\varepsilon \to 0} \limsup_{n \to \infty} \frac{1}{n} D_H^{\varepsilon}(\S_n\|\T_n) \leq \limsup_{n \to \infty}\min_{\substack{\E_n \in \C(\S_n)\\ \F_n \in \C(T_n)}}  {1 \over n} D(\E_n\|\F_n) \label{eq:channel_converse_limsup}\,.
		\end{equation}
		Combining \eqref{eq:channel_converse_liminf} with the achievability result \eqref{eq:channel_achievability_liminf}, we find
		\begin{equation}
			\liminf_{n \to \infty} \min_{\substack{\E_n \in \C(\S_n)\\ \F_n \in \C(T_n)}}  {1 \over n} D(\E_n\|\F_n) \geq \lim_{\varepsilon \to 0} \liminf_{n \to \infty} \frac{1}{n} D_H^{\varepsilon}(\S_n\|\T_n) \geq  \limsup_{k \to \infty} \min_{\substack{\E_k \in \C(\S_k)\\ \F_k \in \C(T_k)}}  {1 \over k} D(\E_k\|\F_k) 
		\end{equation}
		and hence both inequalities in this line are in fact equalities. Also, combining this again with \eqref{eq:channel_converse_liminf} and \eqref{eq:channel_converse_limsup} we find
		\begin{align}
			\lim_{k \to \infty}\min_{\substack{\E_k \in \C(\S_k)\\ \F_k \in \C(T_k)}}  {1 \over k} D(\E_k\|\F_k)  &\leq \lim_{\varepsilon \to 0} \liminf_{n \to \infty} \frac{1}{n} D^{\varepsilon}_H(\S_{n}\|\T_{n}) \\ &
			\leq \lim_{\varepsilon \to 0} \limsup_{n \to \infty} \frac{1}{n} D^{\varepsilon}_H(\S_{n}\|\T_{n}) \leq \lim_{n \to \infty} \min_{\substack{\E_n \in \C(\S_n)\\ \F_n \in \C(T_n)}}  {1 \over n} D(\E_n\|\F_n)\,,
		\end{align}
		and hence all of these expressions coincide, and we get the desired statement with $\liminfsup$. 

  Finally, the second part of the theorem that applies if $\S_n$ also lies in a linear space with dimension polynomial in $n$, can be seen as an immediate consequence of the first part of the theorem after using \autoref{prop:main_exchange} and \autoref{lem:entropy_convex_caratheodory}. Note that after the application of \autoref{lem:entropy_convex_caratheodory} we do no longer satisfy the convexity assumption necessary for another application of \autoref{prop:main_exchange}, and hence we cannot conclude that the min and max can be exchanged again at this point.
	\end{proof}

\subsection{Classical parallel exponent for finite sets in the composite IID setting}

The characterizations of the asymptotic error exponent in \autoref{thm:parr} are generally hard to calculate, and we would like to find scenarios where one can instead find single-letter formulas. One case in which one can do so (and which will be useful for upcoming examples) is in the \emph{classical} composite i.i.d. case when the two sets $\S$ and $\T$ are also assumed to be finite. 
\begin{proposition}\label{prop:classical_finite}
	Let $\S$, $\T \subset \cptp[\X \to \Y]$ be two finite sets of classical channels. In the composite i.i.d setting, i.e. with
	\begin{align}
		\S_n &\coloneqq \Set{\E^{\otimes n} | \E \in \S}\\
		\T_n &\coloneqq \Set{\F^{\otimes n} | \F \in \T}
	\end{align}		
	the Stein exponent of distinguishing these two hypotheses with a parallel strategy is given by
	\begin{equation}
		\lim_{\varepsilon \to 0} \liminfsup_{n \to \infty} e_P(n, \varepsilon, \S_n, \T_n) = \max_{\nu \in \DM[\X' \X]} \min_{\substack{\E \in \S \\ \F \in \T}} D(\E(\nu)\|\F(\nu))\,,
	\end{equation}
	where $\X'$ is another classical system with the size of $\X$, and $\liminfsup$ can be $\liminf$ or $\limsup$ (see \autoref{def:liminfsup}),
\end{proposition}
\begin{proof}~
    \paragraph{Achievability}
    Picking any classical input state $\nu \in \DM[\X' \X]$ and feeding identical copies of it into the $n$ classical channels, turns this problem into the classical composite i.i.d. hypothesis testing problem of distinguishing the sets $P = \S[\nu]$, $Q = \T[\nu]$. Since they are both finite, we can apply \cite[Theorem III.2]{mosonyi_error_2021}, which states that the optimal exponent of this composite state discrimination problem is given by
    \begin{equation}
        \min_{p \in P, q \in Q} D(P\|Q) = \min_{\substack{\E \in \S \\ \F \in \T}} D(\E(\nu)\|\F(\nu))\,.
    \end{equation}
    Taking the supremum over all input states $\nu$ yields the desired achievability result, and the supremum is achieved by an argument similar to \autoref{lem:sup_achieved}, as the minimum over a finite number of elements does not affect any of the required continuity properties. 
    \paragraph{Converse}
    It follows immediately from the definition \eqref{eq:parr_exponent_dh}, \autoref{lem:upper_bound_individual_mmt} and \autoref{lem:wang_renner_upper_bound} that 
    \begin{align}
        \frac{1}{n} D_H^{\varepsilon}(\S_n\|\T_n) &= \frac{1}{n} \sup_{\nu_n \in \DM[R \X^n]} D_H^{\varepsilon}(\S_n[\nu_n]\|\T_n[\nu_n]) 
         \\&\leq {1 \over n(1 - \varepsilon)} \sup_{\nu_n \in \DM[R \X^n]} \inf_{\substack{\E \in \S \\ \F \in \T}} D(\E^{\otimes n}(\nu_n)\|\F^{\otimes n} (\nu_n)) + o(1) \,.\label{eq:classical_finite_converse}
    \end{align}
    where strictly speaking $R$ could be a quantum system, and hence $\nu_n$ a quantum-classical state of the form
     \begin{equation}
         \nu_n = \sum_{i = 1}^{d^n} p_i \rho_R^{(i)} \otimes \ketbra{i}_{\X^n}
     \end{equation}
    where $\{p_i\}_{i = 1}^{d^n}$ is a probability distribution and the $\rho_R^{(i)}$ are density matrices on $R$. 
    
    By using the joint convexity of relative entropy and additivity under tensor products, it is easy to see though that for any such state $\nu_n$ there exists a probability distribution $\{q_i\}_{i = 1}^d$ such that for all channels $\E$ and $\F$:
	\begin{equation}
		{1 \over n} D(\E^{\otimes n}(\nu_n) \| \F^{\otimes n}(\nu_n)) \leq \sum_{i = 1}^d q_i D(\E(\ketbra{i})\|\F(\ketbra{i})) = D(\E(\nu)\|\F(\nu))\,.
	\end{equation}
	where $\nu = \sum_i q_i \ketbra{i}_{\X'} \otimes \ketbra{i}_\X$. 
    Hence, we can upper bound the last part of \eqref{eq:classical_finite_converse} with the single-letter expression

    \begin{equation}
       {1 \over (1 - \varepsilon)} \sup_{\nu \in \DM[\X' \X]} \min_{\substack{\E \in \S \\ \F \in \T}} D(\E(\nu)\|\F(\nu)) + o(1) 
    \end{equation}
    and the statement follows in the limit $n \to \infty$, $\varepsilon \to 0$.
    
\end{proof}

	\subsection{The Adaptive Case}
	\renewcommand{\N}{\mathcal{N}}
	As stated previously, the most general setup of the channels will allow for channel inputs to depend on previous channel outputs, which is called an adaptive protocol.
	Let $n$ be fixed and let $\Lambda_n = \Lambda^{(1)} \otimes \ldots \otimes \Lambda^{(n)}$ be $n$ black-box channels given to us, where the task is to determine whether they come from $\S_n$ or $\T_n$, where $\S_n$ and $\T_n$ are part of quantum channel hypotheses as specified in \autoref{def:composite_channel_hypothesis}. We write just the first $i$ channels as $\Lambda_i \coloneqq \Lambda^{(1)} \otimes \ldots \otimes \Lambda^{(i)}$ for $i = 1, ..., n$.
	A general adaptive channel discrimination protocol for these $\Lambda_n$, can now be fully specified by an initial state $\omega_0  \in \DM[R \otimes A]$, a set of $n - 1$ CPTP maps $\N_i: R \otimes B \to R \otimes A$, that transform the state before it is fed into the next black-box channel, and a final binary POVM $\{M, \IdentityMatrix -  M\}$ on $R \otimes B$. We will assume the size of reference system $R$ to be fixed and identical throughout the protocol (this is without loss of generality). The protocol consists of alternating applications of a black-box channel and the preparation CPTP maps $\N_i$ (see \autoref{fig:adaptive_strategy}). We define:
	\begin{equation}
		\omega_i(\Lambda_i) \coloneqq \Lambda^{(i)}(\N_i (\omega_{i-1}(\Lambda_{i-1}))),  \qquad \text{for } i \in \{2, \ldots, n\},
	\end{equation}
	where we do not make identities on reference systems explicit (as previously), and $\omega_1(\Lambda_1) \coloneqq \Lambda^{(1)}(\omega_0)$.
	With our notation, the final state before the action of the POVM will be $\omega_n(\Lambda_n)$. Note that since the sets $\S_n$ and $\T_n$ were assumed to be permutation invariant, there is no advantage to be gained from reordering the black-box channels and so this is indeed the most general setup. 

	\begin{figure}[htb]
		\centering
		\includegraphics[width=\linewidth]{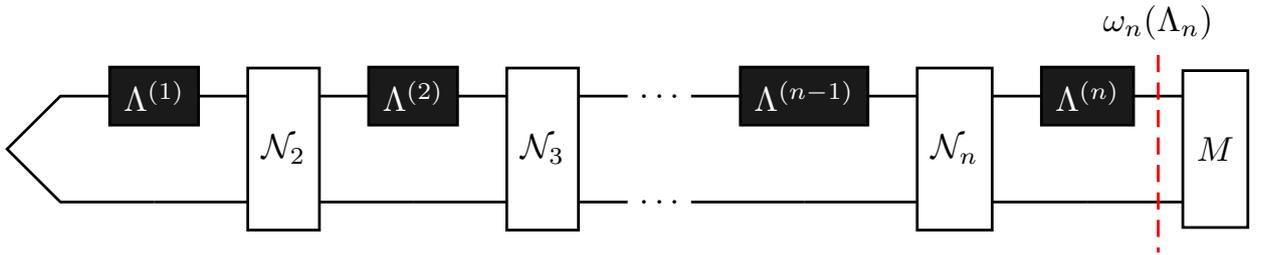}
		\caption{Illustration of a general adaptive protocol with $n$ not necessarily identical black-box channels. The top row makes use of the given black-boxes while the bottom row depicts the memory system $R$. }
		\label{fig:adaptive_strategy}
	\end{figure}

	For a set $\S_n$ corresponding to a hypothesis, we write $\omega_n(\S_n) \coloneqq \Set{\omega_n(\E_n) | \E_n \in \S_n}$. Given an $\omega_n$, the problem then reduces to the composite state-discrimination problem $\omega_n(\S_n)$ vs.\ $\omega_n(\T_n)$. Note that $\omega_n(\S_n) \subset \DM[R \otimes B]$, so this state discrimination problem will not be an instance of a many-copy discrimination problem as studied above, the $n$ just indicates how many channel black-boxes were used in obtaining the states in the set. 
	We can again define the corresponding worst-case type II error exponent as
	\begin{equation}
		e_A(n, \varepsilon, \S_n, \T_n) \coloneqq  \frac{1}{n} \sup_{\omega_n} D_H^{\varepsilon}(\omega_n(\S_n)\|\omega_n(\T_n)) = \frac{1}{n} \sup_{\omega_n} \sup_{\substack{0 \leq M \leq \IdentityMatrix \\ \alpha(M, \omega_n(\S_n)) \leq \varepsilon}} \qty[- \log \beta(M, \omega_n(\T_n))]
	\end{equation}
where the supremum over $\omega_n$ goes over all adaptive strategies, i.e.~all initial states $\omega_0$ and all preparation maps $\N_i$, $i = 2, ..., n$. 
	
	\subsubsection{An upper bound for adaptive strategies}
	We can prove the following upper bound on the Stein exponent for discriminating two composite channel hypotheses with adaptive strategies. This captures the intuition that if the sets $\S_n$ and $\T_n$ are such that they include the i.i.d. problem, then the error exponent has to be less than the worst-case i.i.d. error exponent (for a similar statement for composite state discrimination, see e.g. \cite{mosonyi_error_2022}).
	 \begin{proposition}\label{prop:adaptive_upper_bound}
        Let $\bm{\S} = (\S_n)_n, \bm{\T} = (\T_n)_n$ be two quantum composite channel hypotheses (as in \autoref{def:composite_channel_hypothesis}). Then, for all $n$ and $ε \in [0,1]$, it holds that
        \begin{equation}
            e_A(n, ε, \S_n, \T_n) \leq \inf_{\substack{\E_n \in \S_n\\ \F_n \in \T_n}} e_A(n, ε, \E_n, \F_n)\,.
            \end{equation}
            Furthermore, let $\S \coloneqq \S_1$ and $\T \coloneqq \T_1$. If the hypotheses are such that for all $n$ 		\begin{align}
			\E^{\otimes n} &\in \S_n \quad \forall \, \E \in \S \\
			\F^{\otimes n} &\in \T_n \quad \forall \, \F \in \T
		\end{align}		
		then the Stein exponent for distinguishing these two composite hypotheses by an adaptive strategy is upper bounded by
		\begin{equation}
			\lim_{\varepsilon \to 0} \limsup_{n \to \infty} e_A(n, \varepsilon, \S_n, \T_n) \leq \min_{\substack{\E \in \S \\ \F \in \T}} D_A(\E\|\F) = \min_{\substack{\E \in \S \\ \F \in \T}} D_{\mathrm{reg}}(\E\|\F)\,.
		\end{equation}
\end{proposition}
    \begin{proof}
		As mentioned above, we write $\omega_n(\Lambda_n)$ for the state at the end of an adaptive strategy $ω_n$ with $n$ channel uses, when the $n$ black-box channels are given by $\Lambda_n$. With \autoref{lem:upper_bound_individual_mmt} we get:
        \begin{align}
            e_A(n, \varepsilon, \S_n, \T_n) &= \frac{1}{n} \sup_{\omega_n} D_H^{\varepsilon}(\omega_n(\S_n)\|\omega_n(\T_n)) \leq \frac{1}{n} \sup_{\omega_n} \inf_{\substack{ρ_n \in \omega_n(\S_n) \\ σ_n \in \omega_n(\T_n)}} D_H^{ε}(ρ_n\|σ_n) \\
            &= \frac{1}{n} \sup_{\omega_n} \inf_{\substack{\E_n \in \S_n \\ \F_n \in \T_n}} D_H^{ε}(\omega_n(\E_n)\|ω_n(\F_n)) \leq \frac{1}{n}\inf_{\substack{\E_n \in \S_n \\ \F_n \in \T_n}}  \sup_{\omega_n}  D_H^{ε}(\omega_n(\E_n)\|ω_n(\F_n)) \\
            &= \inf_{\substack{\E_n \in \S_n \\ \F_n \in \T_n}} e_A(n, ε, \E_n, \F_n)
        \end{align}
        which proves the first claim. For the second claim, if the requirements are satisfied, we get
        \begin{equation}
            \inf_{\substack{\E_n \in \S_n \\ \F_n \in \T_n}} e_A(n, ε, \E_n, \F_n) \leq \inf_{\substack{\E \in \S \\ \F \in \T}} e_A(n, ε, \E\n, \F\n)\,.
        \end{equation} 
        The known characterization of the adaptive asymptotic error exponent of simple i.i.d.\ channel discrimination (see e.g. \cite{fang_chain_2020, wang_resource_2019}) thus implies
        \begin{equation}
            \lim_{\varepsilon \to 0} \limsup_{n \to \infty} e_A(n, \varepsilon, \S_n, \T_n) \leq \inf_{\substack{\E \in \S \\ \F \in \T}} D^A(\E\|\F) = \inf_{\substack{\E \in \S \\ \F \in \T}} D^{\mathrm{reg}}(\E\|\F)\,.
        \end{equation}
         Finally, the sequence 
    $D(\E^{\otimes n} \|\F^{\otimes n})$ is superadditive in $n$, and hence is monotonically increasing in $n$, and hence by Fekete's Lemma we can replace the limit $n \to \infty$ in the regularized divergence with a supremum over $n$. Thus, the regularized divergence is lower semi-continuous (as the supremum of lower semi-continuous functions is lower semi-continuous), and hence the infimum is achieved (and the same also for the infimum in $D^A$, since $D^A = D^{\mathrm{reg}}$). 
	\end{proof}	
	We will give a classical example in the next section where this upper bound is achieved and is strictly larger than the achievable exponent of parallel strategies. Hence this demonstrates an advantage of adaptive strategies for composite channel discrimination even if everything is classical. 
	\begin{remark}\label{remark:upper_bound_loose}
	While the upcoming example demonstrates that this upper bound can sometimes be achieved, it cannot always be achieved. Hence, it is not a candidate for the optimal asymptotic exponent of adaptive strategies. This can be seen by taking all channels to be replacer channels\footnote{A replacer channel is a quantum channel which outputs a fixed quantum state regardless of the input.}. In this case the task of channel discrimination reduces to that of state discrimination, for which adaptive and parallel strategies are equivalent. In the composite i.i.d. setting (i.e.~when $S_n = \Set{\rho^{\otimes n} | \rho \in S}$ and $T_n = \Set{\sigma^{\otimes n} | \sigma \in T}$) it has been shown that there exist sets $S$ and $T$ such that
	\begin{equation}
		\lim_{\varepsilon \to 0} \liminfsup_{n \to \infty} D_H^{\varepsilon}(S_n\|T_n) = \lim_{n \to \infty} {1 \over n} \inf_{\substack{\rho_n \in \C(S_n) \\ \sigma_n \in \C(T_n)}} D(\rho_n \| \sigma_n) < \inf_{\substack{\rho \in S \\ \sigma \in T}} D(\rho\|\sigma),
	\end{equation}
	and different examples exist where $S$ and $T$ are either convex~\cite[Section 4.2]{berta_composite_2021} or discrete~\cite[Section IV.A]{mosonyi_error_2021}. 
 \end{remark}

\subsubsection{A classical example of an adaptive advantage}
 In the following we give a fully classical example that demonstrates how adaptive strategies can be (also asymptotically) beneficial with composite hypotheses in the composite i.i.d. setting.

\medskip

	\begin{example}\label{example:adaptive_advantage}
    There exist classical composite channel hypotheses $\S = \{\E_1, \E_2\}$ and $\T = \{\F_1, \F_2\}$, such that the adaptive error exponent in the composite i.i.d. setting is strictly larger than the parallel one. Specifically, we show that 
    \begin{equation}
        \lim_{\varepsilon \to 0} \liminfsup_{n \to \infty} {1 \over n} e_A(n, \varepsilon, \S_n, \T_n) = \min_{i,j \in \{1,2\}} D(\E_i\|\F_j) = 2 \lim_{\varepsilon \to 0} \liminfsup_{n \to \infty} {1 \over n} e_P(n, \varepsilon, \S_n, \T_n)
    \end{equation}
    where $\S_n = \Set{\E_i^{\otimes n}|i = 1,2}$, $\T_n = \Set{\F_i^{\otimes n}|i = 1,2}$.
    \end{example}
 
	When defining the channels, we will use quantum notation for convenience, but everything should be seen as classical, i.e.~all states are diagonal in the computational basis.

	The channels used in our example are then: 
	\begin{align}
		\E_1(\rho) &= \tau \otimes \ketbra{0}{0} \\
		\E_2(\rho) &= \tau \otimes \ketbra{1}{1} \\
		\F_1(\rho) &= \frac{1}{2}\left[ \tau + \bkbraket{0|\rho|0} \ketbra{0}{0} + \bkbraket{1|\rho|1} \tau\right] \otimes \ketbra{0}{0} \\
		\F_2(\rho) &= \frac{1}{2}\left[ \tau + \bkbraket{0|\rho|0} \tau + \bkbraket{1|\rho|1} \ketbra{0}{0} \right] \otimes \ketbra{1}{1}
	\end{align}
	Where $\tau = \IdentityMatrix_{2}/2$ is the maximally mixed state. 
	%When inputting states into these channels, we will sometimes write $0 = \ketbra{0}$ and $1 = \ketbra{1}$, i.e.~ 
 For notational simplicity, we denote $\E(0) \coloneqq \E(\ketbra{0})$, $\E(1) \coloneqq \E(\ketbra{1})$.
	\paragraph{The adaptive strategy}
	The channels are constructed to allow for the following adaptive strategy: Given a black-box channel, we first use it with an arbitrary input state. Depending on the second output bit we will be able to determine with certainty the \enquote{index} of the channel, i.e.~we will know that the channel is either $\E_1$ or $\F_1$ if the second bit is zero, or alternatively if the second bit is one we will know that the channel is either $\E_2$ or $\F_2$. It is easy to see that the optimal input state to discriminate $\E_1$ from $\F_1$ is $\ketbra{0}$, whereas the optimal input state to discriminate $\E_2$ from $\F_2$ is $\ketbra{1}$. Hence, in our adaptive strategy, for all subsequent channel uses, we input the value of the second bit we received out of the first channel use. This will lead to the following exponent:
	\begin{equation}\label{eq:classical_example_adaptive_rate}
		\min_{i \in\{ 1,2 \}} \,\max_{\rho \in \DM[\X]} D(\E_i(\rho) \| \F_i(\rho)) = D\big(\E_1(0) \|\F_1(0)\big) = D\big(\E_2(1) \|\F_2(1)\big) = \log_2(4/3)/2\,.
	\end{equation}
 It is easy to see that this is also equal to 
 \begin{equation}
     \min_{i,j \in\{ 1,2 \}} \,\max_{\rho \in \DM[\X]} D(\E_i(\rho) \| \F_i(\rho))
 \end{equation}
 since this minimum is always achieved for $i = j$, as otherwise the second output bit allows for the two channels to be distinguished with certainty, which makes the relative entropy infinite. Since this is equal to the upper bound from \autoref{prop:adaptive_upper_bound} (for classical channels the regularized channel divergence collapses to the single-letter channel divergence), this is an asymptotically optimal adaptive strategy.
	\paragraph{The best parallel strategy}
    By \autoref{prop:classical_finite}, the optimal parallel exponent is given by
    \begin{equation}
		\max_{\nu \in \DM[\X' \X]} \min_{\substack{\E \in \S \\ \F \in T }} D(\E(\nu)\|\F(\nu)) \,.
	\end{equation} Similarly to the argument used in the proof of \autoref{prop:classical_finite}, by using the joint convexity of the relative entropy we find that for any state $\nu \in \DM[\X' \X]$ there exists a $p \in [0,1]$ such that for any $i,j$:
 \begin{equation}
				 D(\E_i(\nu)\|\F_j(\nu)) \leq p D(\E_i(0)\|\F_j(0)) + (1-p) D(\E_i(1)\| F_j(1))\,,
	\end{equation}
 and picking $\nu = p \ketbra{00}_{\X'\X} + (1-p) \ketbra{11}_{\X'\X}$ achieves the right-hand side.
 Hence, we can write:
 \begin{equation}
		\max_{\nu \in \DM[\X' \X]} \min_{\substack{\E \in \S \\ \F \in T }} D(\E(\nu)\|\F(\nu)) = \max_{0 \leq p \leq 1} \, \min_{i,j \in \{1,2\}} \big( p D(\E_i(0)\|\F_j(0)) + (1-p) D(\E_i(1)\| F_j(1))\big)\,.
	\end{equation}
 Similarly to above, the minimum will be achieved at $i = j$, and it is easy to see by explicit computation that the optimum value of $p$ is $1/2$. Since $D(\E_2(0)\|\F_2(0)) = D(\E_1(1)\|F_1(1)) = 0$ the parallel exponent is thus
	\begin{equation}
		\frac{1}{2} D\big(\E_1(0) \|\F_1(0)\big)
	\end{equation}
	which is half the exponent we were able to achieve with the adaptive strategy. It is also easy to see that a way to achieves this parallel exponent is just to alternate the two input states $0$ and $1$. This captures the intuition that since we do not know the \enquote{index} of the channel in advance, we have to balance between the two optimal input states, and half of the time we will have chosen the wrong one, which means that half the channel outputs will be useless, and hence we can only achieve half the rate.
	
	\subsubsection{Classical equality under convexity}
	Looking back at the previous example, one finds that the advantage of the adaptive strategy can be seen as coming from the fact that the order of the maximum over input states and minimum over channels (for example in \eqref{eq:classical_example_adaptive_rate}) matters: The parallel strategy has to find a good input state for all channels (this corresponds to taking the maximum over states outside), whereas the adaptive strategy can reduce the problem to a simple discrimination problem between just two channels and then tailor the input state to these two channels (this corresponds to taking the infimum over channels outside). Indeed, one also finds that an application of our exchange result \autoref{prop:main_exchange} (or similar minimax theorems) is not permitted in this example, as the sets of channels $\S$ and $\T$ are not convex. We show subsequently that, in the classical case, convexity of these sets is indeed sufficient for there not to be an advantage of adaptive strategies. 
	\begin{theorem}\label{thm:classical_equality}
		Let $\bm{\S} = (\S_n \subset \cptp[\X^n \to \Y^n])_n, \bm{\T} = (\T_n \subset \cptp[\X^n \to \Y^n])_n$ be two composite classical channel hypotheses (still satisfying the properties of \autoref{def:composite_channel_hypothesis}). If $\S \coloneqq \S_1$ and $\T \coloneqq \T_1$ are convex, and additionally for all $n$ 
		\begin{align}
			\E^{\otimes n} &\in \S_n \quad \forall \E \in \S \\
			\F^{\otimes n} &\in \T_n \quad \forall \F \in \T
		\end{align}		
		then the Stein exponent of distinguishing these two composite hypotheses (with a possibly adaptive strategy) is given by
		\begin{equation}
			\lim_{\varepsilon \to 0} \liminfsup_{n \to \infty} e(n, \varepsilon, \S_n, \T_n) = \min_{\substack{\E \in \S \\ \F \in \T}} \max_{\nu \in \DM[\X]} D(\E(\nu)\|\F(\nu))
		\end{equation}
		where $\liminfsup$ can be $\liminf$ or $\limsup$ (see \autoref{def:liminfsup}), and this optimal exponent can be achieved with a parallel strategy.
	\end{theorem}
	\begin{proof} We split the proof in to the achievability and converse parts.
		\paragraph{Achievability}
		Picking any classical input state $\nu \in \DM[\X]$ and feeding identical copies of it into the $n$ classical channels, turns this problem into the classical composite hypothesis testing problem which is at most as hard as distinguishing the sets $P = \S[\nu]$, $Q = \T[\nu]$ in an adversarial setting (this follows from the properties of a composite channel hypothesis as specified in \autoref{def:composite_channel_hypothesis}). Then, by \autoref{thm:adversarial}, the exponent 
		\begin{equation}
			\min_{p \in P, q\in Q} D(p\|q) 
		\end{equation}
		is achievable, and hence, by optimizing over $\nu$, also the exponent
		\begin{equation}
			\sup_{\nu \in \DM[\X]} \min_{\substack{p \in \S[\nu] \\ q \in \T[\nu]}} D(p\|q) = \sup_{\nu \in \DM[\X]} \min_{\substack{\E \in \S \\ \F \in \T}} D(\E(\nu)\|\F(\nu))
		\end{equation}
		is achievable. Now, since $\S$ and $\T$ are convex, we can apply \autoref{prop:main_exchange} and exchange the minimum and the supremum (where the supremum is also achieved, e.g.~by \autoref{lem:sup_achieved}). %This is in fact the only place where the convexity is required. 
		
		\paragraph{Converse}
		From \autoref{prop:adaptive_upper_bound} we get:
		\begin{equation}
			\lim_{\varepsilon \to 0} \lim_{n \to \infty} e(n, \varepsilon, \S_n, \T_n) \leq \min_{\substack{\E \in \S \\ \F \in \T}} D_{\mathrm{reg}}(\E\|\F)\,.
		\end{equation}
		If all channels $\E$ and $\F$ are classical, the regularization is not necessary~\cite{hayashi_discrimination_2009}. This can be easily seen as follows: Since the relative entropy is jointly convex, the optimization over the input state is achieved at an extreme point of the convex set of input states, and classically all extreme points are product distributions, which makes the regularization collapse and also eliminates the need for any reference system. 
		Hence
		\begin{equation}
			\lim_{\varepsilon \to 0} \limsup_{n \to \infty} e(n, \varepsilon, \S_n, \T_n) \leq \min_{\substack{\E \in \S \\ \F \in \T}} \max_{\nu \in \DM[\X]} D(\E(\nu)\|\F(\nu))
		\end{equation}
		which is what we wanted to prove. 
		
	\end{proof}

	\section{Open Problems}
    We have been able to provide new insight into the relation between adaptive and parallel channel discrimination strategies, by studying such strategies for composite channel hypotheses and demonstrating that there is a gap in the asymptotic setting. However, there are still many open questions regarding composite channel discrimination, as can be seen by the number of cells in \autoref{tab:main_results} for which we cannot give a definitive answer. Here, we want to briefly describe some of these problems and elaborate on possible solutions. 

    First of all, for classical composite hypotheses which are non-convex, we currently do not have an entropic expression for the optimal achievable rate of adaptive strategies, so far we have not even been able to prove that the worst-case i.i.d. upper bound cannot always be achieved\footnote{\cite{mosonyi_error_2021} provide an example where discriminating states is not possible with this upper bound. This, however, requires continuous probability distributions (i.e. the analogue of infinite-dimensional Hilbert spaces), which we do not consider here.}. Intuitively though, we consider it to be unlikely that this bound is always achieved, and we are also not particularly hopeful that there will be a simple entropic formula for the adaptive exponent. This comes from imagining generalizations of \autoref{example:adaptive_advantage}: In our example, determining the index of the channel within the two sets was possible perfectly after only one use, and hence one was able to use the optimal input state for all subsequent channel uses. One could, however, think about examples where determining this index is not perfectly possible, and hence one is expected to have to pay a certain (asymptotically non-vanishing) number of channel uses to distinguish the individual elements of the sets and then prepare the best input state, which should make the upper bound of \autoref{prop:adaptive_upper_bound} not achievable in this case. This procedure of determining which channels in the set we seem to be provided with also becomes significantly more complex once one stops having the symmetry between the sets $\S$ and $\T$ which we have in \autoref{example:adaptive_advantage}, and in the general case it is not obvious at all how one could capture in a simple entropic expression the intricacies of gaining knowledge about which elements in this set one might be given.

    Additionally, we would like to see if there is an advantage for adaptive strategies in the quantum composite i.i.d.~case when the sets of channels $\S_1$ and $\T_1$ are convex (recall that we showed that this is not possible classically). Given that the regularization is necessary in general in the quantum case, and the sets $\S_n$ and $\T_n$ will not be convex even if $\S_1$ and $\T_1$ are, we consider it not unlikely that there will again be an asymptotic gap between adaptive and parallel strategies. 

    Finally, we have only studied asymmetric error exponents in this work, even though it would of course also be very interesting to look at similar problems for symmetric error exponents and potentially also Hoeffding exponents.    
    
	\printbibliography
	
	\appendix
	\section{Technical Lemmas}
 Lemma~\ref{lem:wang_renner_upper_bound} is a well-known consequence of the data-processing inequality of the relative entropy and has been widely used in converse proofs in information theory. A statement and proof using our notation can be found in~\cite{wang_one-shot_2012}, although the essence of the statement can already be found much earlier, for example in~\cite[Theorem 2.2]{hiai_proper_1991} and also~\cite[(3.30)]{hayashi_quantum_2006}.
    \begin{lemma}[Upper bound on $D_H^\varepsilon$]\label{lem:wang_renner_upper_bound}
	Let $\rho, \sigma \in \DM$ be quantum states. Then for all $\varepsilon \in [0, 1)$ 
	\begin{equation}
		D_H^{\varepsilon}(\rho\|\sigma) \leq \frac{1}{1 - \varepsilon}\big(D(\rho\|\sigma) + h(\varepsilon)\big),
	\end{equation}
	where $h(\varepsilon)$ is the binary entropy function.
	\end{lemma}
 
	\begin{lemma}\label{lem:upper_bound_individual_mmt}
		For any two sets of states $S$ and $T$, and any $\varepsilon \in [0,1]$,
		\begin{equation}
			D_H^{\varepsilon}(S\|T) \leq \inf_{\substack{\rho \in S \\ \sigma \in T}} D_H^{\varepsilon}(\rho\|\sigma), 
		\end{equation}
		where the intuition is that the left-hand side corresponds to choosing one measurement for all pairings of $\rho$ and $\sigma$, and the right-hand side allows for a different measurement for each pairing.
	\end{lemma}
	\begin{proof}
		One finds that
		\begin{align}
			2^{-D_H^{\varepsilon}(S\|T)} &= \inf_{\substack{0 \leq M \leq \IdentityMatrix \\ \sup_{\rho \in S} \Tr(\bar{M} \rho)\leq \varepsilon}} \sup_{\sigma \in T} \Tr(M \sigma)  
			\geq \sup_{\rho \in S} \inf_{\substack{0 \leq M \leq \IdentityMatrix \\ \Tr(\bar{M} \rho)\leq \varepsilon}} \sup_{\sigma \in T} \Tr(M \sigma) \\
			&\geq \sup_{\rho \in S} \sup_{\sigma \in T} \inf_{\substack{0 \leq M \leq \IdentityMatrix \\ \Tr(\bar{M} \rho)\leq \varepsilon}} \Tr(M \sigma) = \sup_{\substack{\rho \in S \\ \sigma \in T}} 2^{-D^{\varepsilon}_H(\rho\|\sigma)}
		\end{align}
		where we used the notation $\bar{M} = \IdentityMatrix - M$, and the first inequality can be seen as follows: For any $\rho \in S$, if an $M$ is chosen such that $\sup_{\rho' \in S}  \Tr(\bar{M} \rho') \leq \varepsilon$, then obviously also $\Tr(\bar{M} \rho) \leq \varepsilon$ and hence the infimum on the right-hand is over a set of $M$ which can only be larger, and hence the expression can only be smaller. The second inequality is just a very basic property of infima and suprema, and the desired statement then follows by taking negative logarithms.	
	\end{proof}

	\renewcommand{\D}{\mathbf{D}}
		
	The following is a standard argument used to show that the size of reference systems can be restricted. We prove it again here for completeness in a setting that also includes infima over channels.
	\begin{lemma}\label{lem:restrict_reference_size}
		Let $\D$ be a quantum divergence satisfying the data-processing inequality, and let $\S, \T \subset \cptp$ be two arbitrary sets of channels. Then,
		\begin{equation}
			\sup_{\substack{\nu \in \DM[R \otimes A]\\ R \text{ arbitrary}}}	\inf_{\substack{\E \in \S \\ \F \in \T}} \D(\E(\nu) \|\F(\nu)) = \sup_{\substack{\nu \in \DM[R \otimes A] \\ R \cong A}}  \inf_{\substack{\E \in \S \\ \F \in \T}} \D(\E(\nu) \|\F(\nu)),
		\end{equation}
		i.e.~the size of the system $R$ can be restricted to be isomorphic to $A$. 
	\end{lemma}
	
	\begin{proof}
		Let $R$ be an arbitrary reference system, and consider a state $\nu_{RA} \in \DM[R \otimes A]$. Let $\nu$ have a purification $\nu_{PRA}$. Furthermore, take the state $\nu_{A} = \Tr_{R} \nu_{AR}$ and its canonical purification $\nu_{S A}$, where $S \cong A$ (it is well known that such a canonical purification always exists). Then, as also $\nu_{PRA}$ is a purification of $\nu_A$, $\nu_{PRA}$ and $\nu_{SA}$ are related by an isometry $V: PR \to S$. As the channels $\E$ and $\F$ act as identity on these systems the isometry commutes with them, and as any divergence satisfying the data-processing property is invariant under isometries (see e.g.~\cite{khatri_principles_2020})
		\begin{equation}
			D(\E(\nu_{RA})\|\F(\nu_{RA})) \leq D(\E(\nu_{PRA})\|\F(\nu_{PRA}))  = D(\E(\nu_{SA})\|\F(\nu_{SA})))
		\end{equation}
		and hence the supremum can be restricted to reference systems isomorphic to $A$. Note that the set $\D(S \otimes A) = \D(A \otimes A)$ is compact.
	\end{proof}

 	\begin{lemma}[{Generalized minimax theorem~\cite[Theorem 5.2]{farkas_potential_2006}}]
		\label{lem:gen_minimax}
		Let $X$ be a compact and convex subset of a Hausdorff topological vector space and let $Y$ be a convex subset of a linear space. Let $f : X \times Y \to  \reals \cup \{\infty\}$ be lower semi-continuous on $X$ for fixed $y \in Y$, convex in x and concave in y. Then 
		\begin{equation}
			\sup_{y \in Y} \inf_{x \in X} f(x, y) = \inf_{y \in Y} \sup_{x \in X} f(x,y)\,.
		\end{equation}
	\end{lemma}
 \bigskip

	We say that a divergence $\D$ satisfies the direct-sum property, if 
	\begin{equation}
		\D\infdivx*{\bigoplus_{i=1}^n p_i \rho_i }{\bigoplus_{i=1}^n p_i \sigma_i}= \sum_{i=1}^n p_i \D(\rho_i\|\sigma_i)\,.
	\end{equation}
	whenever $\rho_i, \sigma_i \in \HS_i$ are two sets of density matrices and $\{p_i\}_{i = 1}^n$ is a probability distribution.

	\begin{proposition}\label{prop:main_exchange}
		Let $\S,\T \subset \cptp$ be two closed, convex sets of channels. Let $\D$ be a quantum divergence that satisfies the data-processing inequality, is (jointly) lower semi-continuous, and also satisfies the direct-sum property. Then
		\begin{equation}
			\inf_{\substack{\E \in \S \\ \F \in \T}} \sup_{\nu \in \DM[R \otimes A]} \D(\E(\nu)\|\F(\nu)) = \sup_{\nu \in \DM[R \otimes A]} \inf_{\substack{\E \in \S \\ \F \in \T}}  \D(\E(\nu)\|\F(\nu)) .\label{eq:main_exchange}
		\end{equation}
	\end{proposition}

	\begin{proof}
        After publishing the first version of this paper, we became aware that a similar result has also previously been shown in \cite[Theorem 2]{Gour_2019}.
        This proof is inspired by the proof of the similar minimax result \cite[Lemma 13]{brandao_adversarial_2020}.
		The $\geq$ direction follows immediately from very basic properties of $\inf$ and $\sup$. For the $\leq$ direction, let $\mu$ be a discrete measure on the set of density matrices $\DM[RA]$, and consider the function:
		\begin{equation}
			f((\E, \F), \mu) \coloneqq \underset{\nu \sim \mu}{\mathbb{E}} \D(\E(\nu)\|\F(\nu))\,.
		\end{equation}
		where we write $\underset{\nu \sim \mu}{\mathbb{E}}$ for the expectation value with respect to the measure $\mu$ (this is the same as integrating $\nu$ with respect to the measure $\mu$).
		It is easy to see that a divergence that satisfies the data-processing inequality and the direct-sum property is jointly convex (see e.g.~\cite{khatri_principles_2020}). Hence the function $f$ is convex in its first argument, and it is clearly also linear (and hence concave) in the second argument. Note also that the set of channels $\cptp$ is bounded (for example in diamond norm) and hence compact, and so are the closed subsets $\S$ and $\T$. Then, by \autoref{lem:gen_minimax} we have
		\begin{equation}
			\inf_{\substack{\E \in \S \\ \F \in \T}} \sup_{\mu} \underset{\nu \sim \mu}{\mathbb{E}} \D(\E(\nu)\|\F(\nu)) = \sup_{\mu} \inf_{\substack{\E \in \S \\ \F \in \T}} \underset{\nu \sim \mu}{\mathbb{E}}  \D(\E(\nu)\|\F(\nu)).
		\end{equation}
		We can lower bound the left-hand side by restricting the supremum to singular (i.e.~Dirac) measures, which recovers the left-hand side of \eqref{eq:main_exchange}. For the right-hand side, note that by Caratheodory's theorem we can write the expectation value as a convex combination with a finite number of terms. For a given $\mu$ we will write
		\begin{equation}
			\underset{\nu \sim \mu}{\mathbb{E}}  \D(\E(\nu)\|\F(\nu)) = \sum_x p_x \D(\E(\nu_x)\|\F(\nu_x)).
		\end{equation}
		Now define the new state
		\begin{equation}
			\tilde{\nu}_{XRA} \coloneqq \sum_x p_x \ketbra{x}{x} \otimes \nu_x.
		\end{equation}
		By the direct-sum property of $\D$ (note that $\E$ and $\F$ act with an identity on the new system $X$)
		\begin{equation}
			\D(\E(\tilde{\nu})\|\F(\tilde{\nu})) = \sum_x p_x \D(\E(\nu_x)\|\F(\nu_x))\, ,
		\end{equation}
		and so the supremum over measures can be replaced by a supremum over states $\nu_{XRA}$. Finally, by \autoref{lem:restrict_reference_size} the supremum can further be restricted to states $\nu_{RA}$ where $R \cong A$.
	\end{proof}
	It is well-known (see e.g.~\cite{khatri_principles_2020}) that the properties of $\D$ required in \autoref{prop:main_exchange} are satisfied for the Umegaki relative entropy. They are also satisfied for the measured relative entropy:
	\begin{lemma}\label{lem:measured_divergence_properties}
		The measured relative entropy $D_M$ satisfies the data-processing inequality, is lower semi-continuous, and satisfies the direct-sum property. 
	\end{lemma}
	\begin{proof}
		For the data-processing inequality it is easy to see that concatenating a POVM $\M$ (treated as a classical-quantum channel as explained in the mathematical preliminaries) and an arbitrary quantum channel $\E$ as $\M \circ \E$ yields another POVM, and hence $D_M$ satisfies the data-processing inequality as the supremum over all POVMs of the form $\M \circ \E$ can only be smaller than the supremum over all POVMs. $D_M$ is also lower semi-continuous as a supremum of lower semi-continuous functions is lower semi-continuous, and $D$ is lower semi-continuous. For the direct-sum property, let $\HS = \bigoplus_{i = 1}^n \HS_i$, and $\rho_i$, $\sigma_i \in \DM[\HS_i]$ be two sets of density matrices, $\M_i$ be a set of POVMs each on $\HS_i$ and $\{p_i\}_i$ be a probability distribution, where all indices are in the range $i = 1, ... n$. Define $\rho = \bigoplus_{i = 1}^n p_i \rho_i$, $\sigma = \bigoplus_{i = 1}^n p_i \sigma_i$ and $\bar{\M} \coloneqq \bigoplus_{i = 1}^n \M_i$. It is easy to see that 
		\begin{equation}
			\bar{\M}\qty(\bigoplus_{i = 1}^n p_i \rho_i) = \bigoplus_{i = 1}^{n} p_i \M_i(\rho_i)
		\end{equation}
		and hence
		\begin{align}
			D_M\infdivx*{\bigoplus_{i = 1}^n p_i \rho_i}{\bigoplus_{i = 1}^n p_i \sigma_i} &\geq D\infdivx*{\bar{M}\qty(\bigoplus_{i = 1}^n p_i \rho_i)}{\bar{M}\qty(\bigoplus_{i = 1}^n p_i \sigma_i)}\\
			&= D\infdivx*{\bigoplus_{i = 1}^n p_i \M_i(\rho_i)}{\bigoplus_{i = 1}^n p_i \M_i(\sigma_i)} \\
			&= \sum_{i=1}^n p_i D(\M_i(\rho_i)\|\M_i(\sigma_i)),
		\end{align}
		which leads to the $\geq$ direction in the direct-sum property after optimizing over the $\M_i$. For the reverse direction, let $\M$ be an arbitrary POVM on $\HS$. Then
		\begin{equation}
			D(\M(\rho)\|\M(\sigma)) = D\infdivx*{\sum_i p_i \M(\rho_i)}{\sum_i p_i \M(\sigma_i)} \leq \sum_i p_i D(\M(\rho_i)\|\M(\sigma_i)) \leq \sum_i p_i D_M(\rho_i\|\sigma_i),
		\end{equation}
		where we used the joint convexity of the relative entropy. The claim follows by optimizing over all $\M$. 
	\end{proof}

	\begin{lemma}\label{lem:umegaki_continuity}
		Let $\E, \F \in \cptp$ be such that $D_{\max}(\E\|\F) < \infty$, and let $R$ be an arbitrary auxiliary system. Then the function
		\begin{equation}
			(\M, \nu) \mapsto D(\M \circ \E(\nu)\|\M \circ \F(\nu))
		\end{equation}
		is continuous on $\cptp[B \to C] \times \DM[R \otimes A]$.
	\end{lemma}
	\begin{proof}
		It is easy to see that for $0 < \alpha < 1$ the Petz-\renyi divergence \cite{petz_quasi-entropies_1986}
		\begin{equation}
			D_\alpha(\rho\|\sigma) \coloneqq \frac{1}{\alpha - 1} \log \Tr(\rho^{\alpha} \sigma^{1 - \alpha})
		\end{equation}
		is jointly continuous in $\rho$ and $\sigma$. Using the continuity bound from~\cite[Lemma 9]{bergh_parallelization_2022}, we find
		\begin{equation}
			| D_{\alpha}(\rho\|\sigma) - D(\rho\|\sigma) | \leq (1 - \alpha) \log^2\qty(2^{D_{\max}(\rho\|\sigma)} + 2)\,.
		\end{equation}
		As $D_{\max}(\M \circ \E(\nu) \|\M \circ \F(\nu)) \leq D_{\max}(\E\|\F)$ which is independent of $\M$ and $\nu$, we get that $D_{\alpha}(\M \circ \E(\nu)\|\M \circ \F(\nu))$ converges to $D(\M \circ \E(\nu)\|\M \circ \F(\nu))$ uniformly in $\nu$ and $\M$ as $\alpha \uparrow 1$. Hence, also the limiting function $(\M, \nu) \mapsto D(\M \circ \E(\nu)\|\M \circ \F(\nu))$ is continuous.
	\end{proof}
	
	\begin{lemma}\label{lem:sup_achieved}
		Let $\E, \F \in \cptp$, and $R$ be an arbitrary auxiliary system. Then, if $\D$ is either $D$ or $D_M$:
		\begin{equation}
			\sup_{\nu \in \DM[R \otimes A]} \D(\E(\nu) \|\F(\nu)) = \max_{\substack{\nu \in \DM[R \otimes A] \\ R \cong A}} \D(\E(\nu) \|\F(\nu)) 
		\end{equation}
		i.e.~the supremum is achieved and $R$ can be chosen isomorphic to $A$. 
	\end{lemma}	
	\begin{proof}
		The restriction to $R$ isomorphic to $A$ follows from \autoref{lem:restrict_reference_size}. What remains to show is that the supremum is achieved.
		Consider first the case where $D_{\max}(\E\|\F) = \infty$. As shown in~\cite{wilde_amortized_2020}, $D_{\max}(\E\|\F) = D_{\max}(\E(\Omega)\|\F(\Omega))$ where $\Omega = \Omega_{RA}$ is a maximally entangled state. It is well known that (in finite dimensions) $D_{\max}(\rho\|\sigma) = \infty \Rightarrow D(\rho\|\sigma) = \infty$ for all states $\rho$, $\sigma$, and hence if $\D = D$ the value of infinity is achieved in this case. It is also not hard to see that if $D(\rho\|\sigma) = \infty$ also $D_M(\rho\|\sigma) = \infty$ (in this case $\sigma$ will have a smaller support than $\rho$, so a POVM built from the projection onto the support of $\sigma$ will achieve infinity). Hence in this case also with $\D = D_M$ the supremum is achieved.
		So assume $D_{\max}(\E\|\F) < \infty$. If $\D = D$, then by \autoref{lem:umegaki_continuity}, the function is continuous in $\nu$ and hence, since the set of density matrices optimized over is compact, the optimum value is achieved. If $\D = D_M$ we have the expression 
		\begin{equation}
			\sup_{\nu \in \DM[R \otimes A]} D_M(\E(\nu) \|\F(\nu)) = \sup_{\nu \in \DM[R \otimes A]} \sup_{\M \text{ POVM}} D(\M \circ \E(\nu)\|\M \circ \F(\nu))
		\end{equation}
		Again, by \autoref{lem:umegaki_continuity} the expression optimized over is continuous in $\nu$ and $\M$. Also, by~\cite{berta_variational_2017} the optimization over $\M$ in the measured relative entropy can be restricted to von Neumann measurements (i.e.~projective rank-1 measurements), and the set of these measurements is compact (as such a measurement is uniquely specified by a unitary matrix). Hence, also here the optimum value is achieved. 
	\end{proof}
	We say that a state $\rho_n \in \DM[\HS^{\otimes n}]$ is \emph{permutation invariant}, if for any permutation of $n$ systems $\pi \in \mathfrak{S}(n)$ and associated unitary operator $P_{\HS}(\pi)$ it holds that $P_{\HS}(\pi) \rho_n P_{\HS}(\pi)^\dagger = \rho_n$. We say that a channel $\E_n \in \cptp[A^n \to B^n]$ is \emph{permutation covariant} if $\E_n(P_A(\pi) \rho_n P_A(\pi)^\dagger) = P_B(\pi) \E_n(\rho_n) P_B(\pi)^\dagger$ for all input states $\rho_n \in \DM[A^n]$ and all permutations $\pi \in \Sn$. We say that a set of channels $\S_n \subset \cptp[A^n \to B^n]$ is \emph{closed under permutations}, if for any $\E_n \in \S_n$ and any permutation $\pi \in \Sn$, also the channel with permuted input and output systems $\rho \mapsto P_B(\pi)^\dagger \E_n\big(P_A(\pi) \rho P_A(\pi)^\dagger\big) P_B(\pi)$ is an element of $\S_n$.

The simplest examples of permutation invariant states are just tensor power states, although the set of permutation invariant states is significantly bigger than that. Similarly, the simplest examples of permutation covariant channels are channels which are tensor powers, i.e. $\E_n = \E\n$, although again the set of permutation covariant channels is significantly bigger than that. The main reason for permutation invariance being important in this thesis is the following Lemma, which establishes that if a sequence of channels is permutation covariant (e.g.\ because it is an i.i.d.\ string of channels), then also the optimizing input state will be permutation invariant. 

\begin{lemma}\label{lem:inf_perm_covariant}
		Let $\S_n$, $\T_n \subset \cptp[A^n \to B^n]$ be closed convex and also closed under permutations, and let $\D$ be lower semi-continuous and satisfy the data-processing inequality. Then,
		\begin{equation}
			\inf_{\substack{\E_n \in \S_n \\ \F_n  \in \T_n}} \sup_{\substack{\nu \in \DM[R \otimes A^{\otimes n}] \\ R \text{ arbitrary}}} \D(\E_n(\nu) \|\F_n(\nu)) = 	\min_{\substack{\E_n \in \S_n \\ \F_n  \in T_n \\ \E_n, \F_n \text{ perm. covariant.}}} \sup_{\substack{\nu \in \DM[R \otimes A^{\otimes n}] \\ R \text{ arbitrary}}} \D(\E_n(\nu) \|\F_n(\nu))
		\end{equation}
		i.e.~the infimum is achieved for permutation covariant elements of $\S_n$ and $\T_n$. 
	\end{lemma}
	\begin{proof} 
		First, the infimum is achieved since the infimum of a lower semi-continuous function over a compact set is achieved, and the supremum of multiple lower semi-continuous functions is lower semi-continuous. Let $\E_n \in \S_n$ and $\F_n \in \T_n$ be two channels, and consider the permuted versions:
		\begin{align}
			\bar{\E}_n(\rho) &\coloneqq \frac{1}{n!} \sum_{\pi \in \Sn} \Pbd \E_n(\Pa \rho \Pad) \Pb \\
			\bar{\F}_n(\rho) &\coloneqq \frac{1}{n!} \sum_{\pi \in \Sn} \Pbd \F_n(\Pa \rho \Pad) \Pb
		\end{align}
		These two permuted versions can also be seen as a permutation super-channel having been applied to $\E_n$ and $\F_n$, and it is known that any channel divergence can only decrease under the action of such super-channels (see e.g.~\cite{gour_comparison_2019}), however, we will still show this explicitly again here for the reader's convenience:
		
		Define the channel $\mathcal{A} : A^n \to A^n \otimes R'$
		\begin{equation}
			\mathcal{A}(\rho) \coloneqq \frac{1}{n!} \sum_{\pi \in \Sn} (\Pa \rho \Pad) \otimes \ketbra{\pi}_{R'}
		\end{equation}
		where $R'$ is some additional classical register storing the permutation $\pi$. Defining $\mathcal{B}: B^n \otimes R' \to B^n$:
		\begin{equation}
			\mathcal{B}(\rho) \coloneqq \sum_{\pi \in \Sn} \Pbd \bra{\pi}_{R'}\rho \ket{\pi}_{R'} \Pb
		\end{equation} 
		we find that
		\begin{align}
			\bar{\E}_n = \mathcal{B} \circ (\E_n \otimes \id_{R'}) \circ \mathcal{A} \\
			\bar{\F}_n = \mathcal{B} \circ (\F_n \otimes \id_{R'}) \circ \mathcal{A}\,.
		\end{align}
		Hence
		\begin{equation}
			\sup_{\substack{\nu \in \DM[R \otimes A^{\otimes n}] \\ R \text{ arbitrary}}} \D(\bar{\E}_n(\nu) \|\bar{\F}_n(\nu)) \leq \sup_{\substack{\nu \in \DM[R \otimes A^{\otimes n}] \\ R \text{ arbitrary}}} \D(\E_n(\mathcal{A}(\nu)) \|\F_n(\mathcal{A}(\nu))) \leq \sup_{\substack{\nu \in \DM[R \otimes A^{\otimes n}] \\ R \text{ arbitrary}}} \D(\E_n(\nu) \|\F_n(\nu)) 
		\end{equation}
		where the first inequality is the normal data-processing inequality (we omitted the $\id_{R'}$), and the second inequality uses that all the states $\mathcal{A}(\nu)$ are itself included in the supremum. As it is easy to see that the channels $\bar{\E}_n$ and $\bar{\F}_n$ are permutation covariant and are also included in $\S_n$ and $\T_n$ (as they were assumed to be closed under permutations), the minimum will be achieved for permutation covariant channels. 
	\end{proof}

\begin{lemma}\label{lem:channel_div_perm_invariant}
    Let $\E_n, \F_n \in \cptp[A^n \to B^n]$ both be permutation covariant and let $\D$ be $D$ or $D_M$. Then, 
    \begin{equation}
        \sup_{\substack{\nu \in \DM[R \otimes A^{\otimes n}] \\ R \text{ arbitrary}}} \D(\E_n(\nu)\|\F_n(\nu)) = 
        \max_{\substack{\nu \in \DM[R^{\otimes n} \otimes A^{\otimes n}] \\ R \cong A \\ \nu \text{ permut. invariant}}} \D(\E_n(\nu) \|\F_n(\nu))
    \end{equation}
    i.e.~the supremum is achieved for a permutation invariant state $\nu_{R^n A^n} = \nu_{(RA)^n}$ where $R$ is isomorphic to $A$. Note that we mean permutation invariant with respect to permutations permuting the n copies of $(RA)$. 
\end{lemma}
\begin{proof}
    This can be seen as a special case of \cite[Proposition II.4]{leditzky_approaches_2018}, although for this special case we provide a slightly simpler proof. We use the above introduced notation for permutations and associated unitary operators, for $\pi \in \Sn$ the action of e.g. $P_A(\pi)$ will only permute the $n$ copies of the system $A$ and ignore any additional (reference) systems. Let $\nu = \nu_{R_0 A^n} \in \DM[R_0 \otimes A^{\otimes n}]$ be an arbitrary state, where $R_0$ is an arbitrary reference system, and let $\pi \in \Sn$. Then by unitary invariance of $\D$ (which follows from the data-processing inequality):
    \begin{align}
        \D(\E_n(\nu)\|\F_n(\nu)) &= \D(\Pb \E_n(\nu) \Pbd \| \Pb \F_n(\nu) \Pbd) \\&= \D\divarg{\E_n\qty(\Pa \nu \Pad)}{\F_n\qty(\Pa \nu \Pad)}\,.
    \end{align}
    Define
    \begin{equation}
        \omega_{P R_0 A^n} = {1 \over n!} \sum_{\pi \in \Sn} \ketbra{\pi} \otimes \Pa \nu_{R_0 A^n} \Pad \,,
    \end{equation}
    where the first register is classical and stores the permutation $\pi$. 
    By the direct sum property we have
    \begin{align}
        \D(\E_n(\nu)\|\F_n(\nu)) &= {1 \over n!} \sum_{\pi \in \Sn} \D\divarg{\E_n\qty(\Pa \nu \Pad)}{\F_n\qty(\Pa \nu \Pad)} \\&= \D(\E_n(\omega_{P R_0 A^n})\| \F_n(\omega_{P R_0 A^n}))\,.
    \end{align}
    Let $\omega_{S P R_0 A^n}$ be a purification of $\omega_{P R_0A^n}$.  Note that 
    \begin{equation}
        \omega_{A^n} = {1 \over n!} \sum_{\pi \in \Sn}\Pa \nu_{A^n} \Pad
    \end{equation}
    is permutation invariant, and hence by~\cite[Lemma II.5]{christandl_one-and--half_2007}, there exists a system $K$ isomorphic to $A$, and a permutation invariant purification $\omega_{(KA)^n} \in \DM[K^{\otimes n} \otimes A^{\otimes n}]$ where the permutations act on $\omega_{(KA)^n}$ by permuting the copies of $(KA)$. Now the two purifications $\omega_{(KA)^n}$ and $\omega_{SPR_0 A^n}$ will be related by a partial isometry $V: K^n \to SPR_0$ which commutes with $\E_n$ and $\F_n$ (since they only act on $A^n$). Hence,
    \begin{align}
        \D(\E_n(\omega_{P R_0 A^n})\| \F_n(\omega_{P R_0 A^n})) &\leq \D(\E_n(\omega_{S P R_0 A^n})\| \F_n(\omega_{S P R_0 A^n})) \\
        &= \D(\E_n(\omega_{(KA)^n})\| \F_n(\omega_{(KA)^n}))
    \end{align}
    by the data processing inequality and isometric invariance. The fact that the supremum is also achieved follows from the same argument as in \autoref{lem:sup_achieved}. 
\end{proof}
The significance of these restrictions to permutation covariant channels and permutation invariant input states comes from the fact that both terms $\E_n(\nu_n)$ and $\F_n(\nu_n)$ will then be permutation invariant, and this allows us to use:
\begin{lemma}[{\cite[Lem. 2.4]{berta_composite_2021}}]\label{lem:measured_asmptotically_equal}
    Let $\rho_n, \sigma_n \in \DM[\HS^{\otimes n}]$ with $\sigma_n$ permutation invariant. Then,
    \begin{equation}
        D(\rho_n\| \sigma_n) - \log \text{\emph{poly}}(n) \leq D_M(\rho_n\|\sigma_n) \leq D(\rho_n\|\sigma_n)\,,
    \end{equation}
\end{lemma}
Additionally, we have the following lemma (essentially a variant of~\cite[Lemma 2.5]{berta_composite_2021}) which allows us to remove a convex hull in the infimum over the first argument, if the states also all lie in a sufficiently small linear subspace. This is for example the case if all states are permutation invariant, as the subspace of permutation invariant density matrices in $\DM[\HS^{\otimes n}]$ lies in a linear subspace of $\BO[\HS\n]$ whith dimension upper bounded by $(n+1)^{(\dim \HS)^2}$.
\begin{lemma}\label{lem:entropy_convex_caratheodory}
    Let $\sigma \in \DM[\HS]$, and let $S \subset \mathcal{W} \cap \DM[\HS]$ be a set of density matrices, where $\mathcal{W}$ is a linear subspace of $\BO[\HS]$. Then,
    \begin{equation}
        \inf_{\rho \in \C(S)} D(\rho \| \sigma) \geq \inf_{\rho \in S} D(\rho\|\sigma) - \log(\dim \mathcal{W} + 1)\,,
    \end{equation}
    where $\C(S)$ is the convex hull of $S$. 
\end{lemma}
\begin{proof}
    By Caratheodory's theorem, we can write any element $\tilde{\rho} \in \C(S)$ as $\sum_{i=1}^{n} p_i \tilde{\rho_i}$, where $p_i$ is a probability distribution, $\tilde{\rho}_i \in S$ and $n = \dim \mathcal{W} + 1$. We can assume $\tilde{ρ}_i \ll σ$ for all $i$, as otherwise $D(\tilde{ρ}\|σ) = \infty$ and there is nothing to show. For $ε > 0$ fixed, let $ρ_i \coloneqq \tilde{ρ}_i + ε \Pi_σ$, where $\Pi_σ$ is the projection onto the support of $σ$. Then,
    \begin{align}
        D\divarg{\sum_i p_i \rho_i}{\sigma} &= \tr[ \qty(\sum_i p_i \rho_i) \qty(\log(\sum_j p_j \rho_j) - \log(\sigma))] \\
        &= \tr[ \sum_i p_i \rho_i \qty(\log(\sum_j p_j \rho_j) - \log(\sigma))] \\
        &\geq \tr[ \sum_i p_i \rho_i \qty(\log(p_i \rho_i) - \log(\sigma))] \\
        &= \sum_i p_i \tr[\rho_i (\log \rho_i - \log \sigma + \log p_i)] \\
        &= \sum_i p_i D(\rho_i\|\sigma) - H(p) \geq \sum_i p_i D(\rho_i\|\sigma)  - \log(n)\,,
    \end{align}
    where for the first inequality we used the operator monotonicity of the logarithm and that $\sum_j p_j \rho_j \geq p_i \rho_i$ for every $i$. Now, by the already mentioned continuity of $D$ in the first variable (when restricted to density matrices on the support of $σ$), we can take the limit $ε \to 0$ on both sides to get
    \begin{equation}
    			D\divarg{\tilde{ρ}}{\sigma} \geq \sum_i p_i D(\tilde{\rho}_i\|\sigma)  - \log(n) \geq \inf_{\rho \in S} D(\rho\|\sigma) - \log(n)\,.
    \end{equation}
\end{proof}
	
\end{document}